\documentclass[copyright,creativecommons]{eptcs}
\providecommand{\event}{GandALF 2026} 

\usepackage{iftex}

\ifpdf
  \usepackage{underscore}         
  \usepackage[T1]{fontenc}        
\else
  \usepackage{breakurl}           
\fi

\title{Non-Wellfounded and Cyclic Proofs for LTL:\\ A Syntactic Correspondence with Linear Nested Sequents (Full Version)}
\author{Tim S. Lyon
\institute{TU Dresden \\ Nöthnitzer Straße 46, 01187 Dresden, Germany}
\email{\quad timothy\_stephen.lyon@tu-dresden.de}
\and
Lukas Zenger\thanks{Supported by the Swiss National Science Foundation, Grant P500-2\_235387.}
\institute{Peking University\\
5 Yiheyuan Rd., Haidian District,
100871 Beijing,
China}
\email{\quad lukas.zenger@pku.edu.cn}
}
\def\titlerunning{Non-Wellfounded and Cyclic Proofs for LTL}
\def\authorrunning{T.S. Lyon \& L. Zenger}

\usepackage{amssymb, amsmath, amsthm}
\usepackage{color}
\usepackage{bussproofs}
\usepackage{proof}
\usepackage{mathtools}
\usepackage{relsize}
\usepackage{xspace}
\usepackage{upgreek}
\usepackage{multicol}
\usepackage{stmaryrd}
\usepackage{mathrsfs}
\usepackage{tikz}
\usepackage{graphicx}
\usetikzlibrary{decorations.pathmorphing}

\usepackage{algpseudocode}
\usepackage[ruled,vlined,linesnumbered]{algorithm2e}
\usepackage{mathdots}

\newtheorem{theorem}{Theorem}[section]
\newtheorem{lemma}[theorem]{Lemma}
\newtheorem{definition}[theorem]{Definition}
\newtheorem{example}[theorem]{Example}
\newtheorem{corollary}[theorem]{Corollary}

\newcommand{\case}[1]{%
  \smallskip
  \noindent\textsc{#1}
}

\newcommand{\fig}{Figure}

\newcommand{\iffi}{iff }
\def\phi{\varphi}

\definecolor{tim}{RGB}{0, 0, 250}

\newcommand{\sufo}[1]{\mathrm{sufo}(#1)}

\newcommand{\imp}{\rightarrow}
\newcommand{\lang}{\mathscr{L}}

\newcommand{\atm}{\mathsf{Atm}}

\newcommand{\X}{\mathtt{X}}

\newcommand{\until}{\mathtt{U}}

\newcommand{\ltl}{\mathsf{LTL}}

\newcommand{\lgth}[1]{\mathsf{len}(#1)}

\newcommand\citet[1]{\citeauthor{#1}~(\citeyear{#1})}

\newcommand{\lnsg}{\mathcal{G}}
\newcommand{\lnsh}{\mathcal{H}}
\newcommand{\lnsk}{\mathcal{K}}
\newcommand{\lnsf}{\mathcal{F}}

\newcommand{\lnsq}{\mathcal{Q}}
\newcommand{\sep}{\sslash}
\newcommand{\id}{\mathsf{id}}
\newcommand{\botl}{\bot\mathsf{L}}

\newcommand{\disr}{\lor \mathsf{R}}
\newcommand{\conr}{\land\mathsf{R}}

\newcommand{\impl}{{\imp}\mathsf{L}}
\newcommand{\impr}{{\imp}\mathsf{R}}
\newcommand{\nextri}{\X \mathsf{R}_{1}}
\newcommand{\nextrii}{\X \mathsf{R}_{2}}
\newcommand{\nextli}{\X \mathsf{L}_{1}}
\newcommand{\nextlii}{\X \mathsf{L}_{2}}

\newcommand{\untilli}{\until \mathsf{L}_{1}}
\newcommand{\untillii}{\until \mathsf{L}_{2}}

\newcommand{\untilri}{\until \mathsf{R}_{1}}
\newcommand{\untilrii}{\until \mathsf{R}_{2}}

\newcommand{\calc}{\mathsf{LNS_{LTL}}}
\newcommand{\pcalc}{\mathsf{LNS^{+}_{LTL}}}
\newcommand{\calci}{\mathsf{LNS^{\infty}_{LTL}}}
\newcommand{\calccyc}{\mathsf{LNS^{cyc}_{LTL}}}
\newcommand{\sts}{\bar{\sigma}}

\newcommand{\sar}{\vdash}

\newcommand{\prf}{\pi}
\newcommand{\fint}{f}

\newcommand{\tri}{\tau}

\newcommand{\tuple}[1]{( #1 )}

\newcommand{\sat}{\vDash}

\newcommand{\der}{\pi} 

\newcommand{\set}[1]{\{#1\}}

\newcommand{\branch}{\mathcal{B}}
\newcommand{\lnspath}{\mathcal{P}}
\newcommand{\dbclo}{\mathsf{DB}}
\newcommand{\cbclo}{\mathsf{CB}}

\newcommand{\comp}[1]{\mathsf{cp}(#1)}

\newcommand{\last}[1]{\mathsf{end}(#1)}

\newcommand{\model}{\bar{\sigma}}

\newcommand{\bracket}[1]{\langle #1 \rangle}

\newcommand{\proves}{\Vdash}

\newcommand{\size}[1]{\left\Vert#1\right\Vert}

\newcommand{\ru}{\mathsf{r}}

\newcommand{\shift}{h} 

\newcommand{\nshole}{\{\cdot\}}

\newcommand{\meq}{=_{M}}

\newcommand{\idxset}{I}

\newcommand{\an}{\mathsf{a}}

\newcommand{\anno}{\mathsf{a}}
\newcommand{\annof}{\mathsf{f}}
\newcommand{\annou}{\mathsf{u}}
\newcommand{\fbr}{\mathsf{fb}}
\newcommand{\ubr}{\mathsf{ub}}

\newtheorem{proposition}[theorem]{Proposition}

\begin{document}
\maketitle

\begin{abstract}
We introduce the formalism of non-wellfounded and cyclic linear nested sequent calculi, developing concrete systems for linear temporal logic ($\ltl$). The paper addresses two central problems, which we call \emph{cycle recognition} and \emph{unraveling}. Cycle recognition concerns identifying cycles in non-wellfounded proofs in order to extract corresponding cyclic proofs, while unraveling studies the converse transformation, from cyclic proofs to non-wellfounded ones. Although these processes are well understood for Gentzen sequents, they have received little attention for more expressive sequent formalisms and become more challenging in the linear nested sequent setting. To address cycle recognition, we show the completeness of non-wellfounded proofs relative to a particular normal form exhibiting a property we call \emph{saturation recurrence}, which enables the systematic extraction of cyclic proofs. To address unraveling, we introduce a specialized procedure that shifts rule applications forward along linear nested sequents, allowing non-wellfounded proofs to be reconstructed from cyclic ones. Overall, our work provides new proof-theoretic techniques for cycle recognition and unraveling in expressive multisequent formalisms.
\end{abstract}

\section{Introduction}


Modal fixed-point and program logics extend modal languages with operators for expressing inductive and coinductive definitions. 
Due to these constructs, such logics often require reasoning principles that go beyond finite, well-founded proofs~\cite{AfsLei17,ForSan13,Min78}. One approach to obtain cut-free and analytic proof systems for such logics is to employ \emph{non-wellfounded proofs}, which take the form of non-wellfounded trees and rely on global soundness conditions to ensure correctness. Such proofs can be viewed as a formalization of reasoning by infinite descent. The relationship between induction, infinite descent, and non-wellfounded proofs has been studied in particular by Brotherston~\cite{Brotherston2006}. He discusses the sense in which induction and infinite descent can be regarded as equivalent proof principles, while emphasizing that the precise correspondence depends on the underlying formalization. When non-wellfounded proofs have the shape of a \emph{regular tree}---containing only finitely many distinct subtrees---they can be folded into finite proofs in which certain leaves are linked back to internal nodes, thereby creating cycles. This gives rise to \emph{cyclic proofs}, which  provide a finite representation of non-wellfounded proofs and are amenable to automated proof-search~\cite{RooZen22} and to computing interpolants~\cite{AfsLei22,Sha14}.

In the setting of \emph{Gentzen sequents} (pairs of (multi)sets of logical formulae), the relationship between non-wellfounded and cyclic proofs is well understood. As alluded to above, if a non-wellfounded proof can be transformed into a regular tree, then it is straightforwardly convertible into a cyclic proof; conversely, one can unravel a cyclic proof into a non-wellfounded proof. By contrast, the correspondence between non-wellfounded and cyclic proofs has received far less attention in the context of \emph{multisequents}.

Multisequents are generalizations of Gentzen sequents, obtained by embedding Gentzen sequents into more complex data structures. Examples include hypersequents (multisets of Gentzen sequents)~\cite{Avr96,Pot83}, linear nested sequents (lines of Gentzen sequents)~\cite{Mas92,Lel15}, and labeled sequents (graphs of Gentzen sequents)~\cite{Sim94,Vig00}. For a survey of multisequent systems and their relationships, see~\cite{LyoEtAl25}. Multisequents were introduced over the past few decades to provide cut-free and analytic sequent-style systems for modal and related logics---particularly for logics where a cut-free Gentzen sequent calculus was not known to exist. 
Multisequent systems often enjoy properties that can make them better suited for proof analysis and applications than traditional Gentzen systems. For instance, many multisequent calculi consist entirely of invertible rules, feature symmetric rules that facilitate cut-elimination, and use structures conducive to counter-model extraction. 

Research on non-wellfounded multisequent calculi remains nascent, with only a handful of works on non-wellfounded hypersequent or labeled sequent systems. A set of non-wellfounded and cyclic hypersequent systems for modal logics with the master modality capturing different frame conditions was provided by Rooduijn~\cite{Roo21,Rooduijn2024}, while Das and Girlando provided a cyclic hypersequent system for transitive closure logic~\cite{DasGir23}. Docherty and Rowe provided a cut-free non-wellfounded labeled system for propositional dynamic logic as well as a cyclic labeled system that includes cut~\cite{DocRow19}.  More recently, Das et al.~\cite{DasGieMar24} presented non-wellfounded labeled calculi for classical and intuitionistic Gödel-Löb logic, with Aguilera and Pacheco subsequently providing a cyclic version of the intuitionistic system~\cite{AguPac25}. Furthermore, Afshari et al.~\cite{AfshariEtAl2023} provided non-wellfounded calculi for intuitionistic temporal logics which employ formula nesting, and their work was improved upon by Men\'endez Turata~\cite{MenendezTurata2024}, who obtained a cut-free labeled cyclic proof system. To the best of our knowledge, for multisequents more complex than hypersequents, the works by Aguilera and Pacheco and by Men\'endez Turata serve as the only examples of cut-free cyclic multisequent systems.


Our aim in this paper is to advance this line of research by providing a deeper understanding of cyclic proofs in the context of multisequents. In particular, we focus on non-wellfoundedness and cyclicity within the \emph{linear nested sequent (LNS)} formalism. Linear nested sequents, which consist of ordered lines of Gentzen sequents, generalize hypersequents. This formalism was introduced by Lellmann~\cite{Lel15} and builds on the 2-sequent framework originally developed by Masini~\cite{Mas92,Mas93}. As a case study, we provide cut-free non-wellfounded and cyclic linear nested sequent calculi for linear temporal logic ($\ltl$), which is our first contribution in this paper. We then study and address two central problems, which we call \emph{cycle recognition} and \emph{unraveling}. The cycle recognition problem asks how to identify cycles in non-wellfounded proofs in order to construct a corresponding cyclic proof, whereas the unraveling problem concerns the reverse process: transforming a cyclic proof into a non-wellfounded one.

Detecting cycles in non-wellfounded multisequent proofs is particularly challenging because multisequents employ richer structures than traditional Gentzen sequents. Repetitions do not recur verbatim along infinite branches; instead, structures tend to grow. Consequently, cycles must be identified not as identical repeated sequents, but as recurring growth patterns. We resolve this problem by showing completeness with respect to non‑wellfounded proofs satisfying the \emph{saturation recurrence property}. This guarantees that repetitive growth patterns can be detected along infinite branches, thereby enabling the extraction of cyclic proofs. This constitutes our second contribution.

Unraveling cyclic proofs in this setting is likewise non-trivial since the standard unraveling procedure is inapplicable due to the more complex structures used in proofs. As a third contribution, we show how this problem can be addressed by introducing a specialized unraveling procedure that shifts rule applications forward along linear nested sequents. This makes it possible to reconstruct a non-wellfounded proof in which linear nested sequents exhibit repetitive growth along infinite branches. Since the non-wellfounded system for $\ltl$ is proven sound and complete, the above two transformations yield soundness and completeness of the cyclic system as a corollary.

\smallskip
\noindent
\textbf{LTL and Related Systems.} Linear temporal logic ($\ltl$) is an important program logic introduced by Pnueli~\cite{Pnu77} for reasoning about dynamically evolving systems. 
The logic is equipped with the temporal operators \emph{next} $\X$ and \emph{until} $\until$, which enable the specification of infinite behaviors over execution traces. This gives 
$\ltl$ considerable expressive power while still admitting effective decision procedures. 
Due to these favorable properties, $\ltl$ is widely used in model checking and formal verification, where properties such as safety (e.g., ‘no bad state will be reached’) and liveness (e.g., ‘a good state will eventually be reached’) can be specified and automatically checked.

A variety of sequent calculi have been proposed for $\ltl$. Paech introduced a Gentzen-style system for the logic~\cite{Pae89}; however, this system contains an induction rule that is an instance of cut, and is therefore not cut-free. To address this issue, Brünnler and Lange~\cite{BruLan08} later developed a cut-free cyclic Gentzen calculus that employs annotations---called \emph{histories}---to detect cycles in proofs. A variant of Brünnler and Lange's calculus was discussed by Kokkinis and Studer~\cite{KokStu16} for the unary fragment of $\ltl$, in which weakening is syntactically admissible. More recently, Alonderis et al.~\cite{AloEtAl20} provided a non-annotated variant of the Brünnler-Lange calculus for full $\ltl$. Boretti~\cite{Bor08} introduced an infinitary labeled system for $\ltl$, which relies on an $\omega$-rule with infinitely many premises, rather than on non-wellfounded proofs, to establish validity. On the satisfiability side, Reynolds~\cite{Rey16} provided a one-pass tree-shaped tableau for $\ltl$ satisfiability that employs the $\mathsf{LOOP}$ and $\mathsf{PRUNE}$ rules to detect satisfiable and unsatisfiable cycles in tableau branches, respectively. While the cycle detection underlying Reynolds' tableau is reminiscent of the cycle recognition studied in the present work, the precise relationship between these approaches remains open.

Since our linear nested sequent systems fall within the paradigm of non-wellfounded proof theory, they are most closely related to~\cite{BruLan08,KokStu16,AloEtAl20}. However, they differ in that we employ the richer formalism of linear nested sequents rather than traditional Gentzen sequents, our calculus is \emph{formula-driven} (i.e., it does not contain any structural rules), and all rules are invertible. We remark that our choice of formalism is motivated by two considerations. First, syntactic cut-elimination remains open for all existing non-wellfounded Gentzen calculi for $\ltl$~\cite{BruLan08,KokStu16,AloEtAl20}, and we conjecture that the additional structural expressiveness of linear nested sequents provides a pathway to resolving this problem. Second, our work contributes to the broader goal of developing a general theory of non-wellfounded and cyclic multisequent systems: the cycle recognition and unraveling techniques introduced here serve as a foundation for extending these methods to other expressive formalisms beyond the linear nested setting.

\smallskip
\noindent
\textbf{Outline of Paper.} In Section~\ref{sec:prelims} we introduce the preliminaries for $\ltl$. Section~\ref{sec:calculi} introduces our non-wellfounded and cyclic linear nested sequent systems for $\ltl$. Section~\ref{sec: soundness} establishes soundness and Section~\ref{sec: completeness} establishes completeness of the non-wellfounded calculus. Section~\ref{sec:correspondence} provides our solutions to the cycle recognition and unraveling problems, establishing a syntactic correspondence between non-wellfounded and cyclic linear nested sequent proofs, which entails the soundness and completeness of the cyclic system. Finally, in Section~\ref{sec:conclusion}, we conclude and discuss future work.

\section{Logical Preliminaries: Linear Temporal Logic}\label{sec:prelims}

Let $\atm := \{p,q,r,\ldots\}$ be a denumerable set of \emph{(propositional) atoms}. We define the language $\lang$ to be the collection of all formulae $A$ generated via the following grammar in BNF:
$$
A, B ::= p \mid \bot \mid (A \imp B) \mid \X A \mid (A \until B)
$$
with $p \in \atm$. We use $A$, $B$, $C$, $\ldots$ to denote formulae from $\lang$. The formula $\X A$ expresses `$A$ is true in the \emph{next} time-step' while $A \until B$ expresses `$A$ is true \emph{until} $B$ is true (and $B$ eventually becomes true)'. We have opted to use a minimal signature for $\ltl$ to simplify the presentation of our linear nested sequent systems in the subsequent section. We define  $\neg A := A \imp \bot$, $\top := \neg \bot$, $A \lor B := \neg A \imp B$, $A \land B := \neg (A \imp \neg B)$, and $A \leftrightarrow B := (A \imp B) \land (B \imp A)$. For a finite set of formulae $\Gamma$, $\bigwedge \Gamma$ and $\bigvee \Gamma$ are defined as usual, where
\begin{equation*}
     \bigwedge \emptyset := \top, \qquad \bigvee \emptyset := \bot.
\end{equation*}
Other temporal operators are defined in terms of $\until$ as follows: $\mathtt{F} A := \top \until A$ expresses `$A$ is \emph{eventually} true', $\mathtt{G} A:= \neg \mathtt{F} \neg A$ expresses `$A$ is \emph{henceforth} true', and $A \mathtt{R} B := \neg (\neg A \until \neg B)$ expresses `$A$ \emph{releases} B'. A \emph{complex formula} is a formula of the form $A \imp B$, $\X A$, or $A \until B$. We let $\sufo{A}$ denote the set of all \emph{subformulae} of $A$, defined in the usual way, and call a formula $B$ a \emph{subformula} of $A$ \iffi $B \in \sufo{A}$. For a set of formulae $\Gamma$, we define $\sufo{\Gamma} := \bigcup_{A \in \Gamma} \sufo{A}$.


\begin{definition}\label{def:semantics} A \emph{state sequence} is an infinite sequence of states $\bar\sigma=\tuple{\sigma_0,\sigma_1,\ldots}$, where each state $\sigma_i \subseteq \atm$. The \emph{satisfaction} of a formula $A$ on $\bar\sigma$ at the time point $i \ge0$, written $\bar\sigma, i \sat A$, is defined by:
\begin{itemize}
\item $\bar\sigma, i \sat p$ \iffi  $p \in \sigma_i$;



\item $\bar\sigma, i \not\sat \bot$;

\item $\bar\sigma, i \sat A \imp B$ \iffi $\bar\sigma, i \not\sat A$ or $\bar\sigma, i \sat B$;


\item $\bar\sigma, i \sat \X A$ \iffi $\bar\sigma, i+1 \sat A$;

\item $\bar\sigma, i \sat A \until B$ \iffi there is a $j \geq i$ such that $\bar\sigma, j \sat B$ and for every $i \leq n < j$, $\bar\sigma, n \sat A$;



\item $\bar\sigma \sat A$ \iffi $\bar\sigma, 0 \sat A$.

\end{itemize}
We say that $A$ is \emph{valid}, written $\sat A$, \iffi for each state sequence $\bar\sigma$ we have $\bar\sigma \sat A$; otherwise, a formula $A$ is \emph{invalid}, written $\not\sat A$. We define the logic $\ltl$ to be the set of all valid formulae from $\lang$.
\end{definition}

\begin{example}
  It is straightforward to verify, using the semantic clauses defined above, that for any formulae $A,B \in \lang$, $A \until B \leftrightarrow B \lor (A \land \X (A \until B))$ is valid. The formula $B \lor (A \land \X (A \until B))$ is called the \emph{unfolding} of $A \until B$. 
\end{example}

\section{Linear Nested Sequents}\label{sec:calculi}

This section introduces linear nested sequents as well as the non-wellfounded and cyclic calculi for $\ltl$. To ensure soundness of the cyclic calculus, we make use of focusing to keep track of good formula traces. 

An \emph{annotated formula} is a tuple $(A, \anno)$ where $A$ is a formula and $\anno \in \{ \annou, \annof \}$ is an \emph{annotation}. The annotation $\annof$ designates that the formula is \emph{in focus} while $\annou$ designates that the formula is \emph{unfocused}. We will write annotated formulae as $A^\anno$ and define $A^{\annou} := A$ for simplicity. That is, we take unfocused formulae to just be regular formulae; this will simplify our presentation. We do not use parentheses when displaying annotated formulae, e.g., $A \until B^\annof$ should be read as $(A \until B)^\annof$.

For a set of annotated formulae $\Gamma$, we define $\Gamma^-:= \{A \mid A^\an \in \Gamma\}$, i.e., $\Gamma^-$ is the set of formulae obtained by removing annotations. Moreover, we define $\Gamma^\annou := \{A^\annou \mid A^\an \in \Gamma\}$ to be the set of annotated formulae obtained from $\Gamma$ by changing any annotation to $\annou$. A \emph{Gentzen sequent} is an expression of the form $\Gamma \sar \Delta$ such that (1) $\Gamma$ and $\Delta$ are finite sets of annotated formulae, (2) at most one formula in $\Gamma$ is in focus, (3) if a formula in $\Gamma$ is in focus, then it is of the form $A \until B$, and (4) every formula in $\Delta$ is unfocused. We call $\Gamma$ the \emph{antecedent} and $\Delta$ the \emph{consequent} of $\Gamma \sar \Delta$.  If $\Gamma \sar \Delta$ contains a formula in focus, then we call $\Gamma \sar \Delta$ \emph{focused}, and otherwise we call  $\Gamma \sar \Delta$ \emph{unfocused}. 

\begin{figure*}[t]
\noindent\hrule

\begin{center}
\begin{tabular}{c c c}
\AxiomC{$\phantom{\lnsg}$}
\RightLabel{$\id$}
\UnaryInfC{$\lnsg \sep \Gamma, p \sar p, \Delta \sep \lnsh$}
\DisplayProof

&

\AxiomC{$\phantom{\lnsg}$}
\RightLabel{$\botl$}
\UnaryInfC{$\lnsg \sep \Gamma, \bot \sar \Delta \sep \lnsh$}
\DisplayProof
\end{tabular}
\end{center}


\begin{center}
\begin{tabular}{c c}
\AxiomC{$\lnsg \sep \Gamma, A \sar  B, \Delta \sep \lnsh$}
\RightLabel{$\impr$}
\UnaryInfC{$\lnsg \sep \Gamma \sar A \imp B, \Delta \sep \lnsh$}
\DisplayProof

&

\AxiomC{$\lnsg \sep \Gamma, B \sar \Delta \sep \lnsh$}
\AxiomC{$\lnsg \sep \Gamma \sar A, \Delta \sep \lnsh$}
\RightLabel{$\impl$}
\BinaryInfC{$\lnsg \sep \Gamma, A \imp B \sar \Delta \sep \lnsh$}
\DisplayProof
\end{tabular}
\end{center}


\begin{center}
\begin{tabular}{c c c}
\AxiomC{$\lnsg \sep \Gamma \sar \Delta \sep \Sigma, A \sar \Pi \sep \lnsh$}
\RightLabel{$\nextli$}
\UnaryInfC{$\lnsg \sep \Gamma, \X A \sar \Delta \sep \Sigma \sar \Pi \sep \lnsh$}
\DisplayProof

&

\AxiomC{$\lnsg \sep \Gamma \sar \Delta \sep A \sar $}
\RightLabel{$\nextlii$}
\UnaryInfC{$\lnsg \sep \Gamma, \X A \sar \Delta$}
\DisplayProof

&

\AxiomC{$\lnsg \sep \Gamma^\annou, A \until B^\annou \sar \Delta \sep \lnsh$}
\RightLabel{$\annou$}
\UnaryInfC{$\lnsg \sep \Gamma^\annou, A \until B^\annof \sar \Delta \sep \lnsh$}
\DisplayProof
\end{tabular}
\end{center}


\begin{center}
\begin{tabular}{c c c}

\AxiomC{$\lnsg \sep \Gamma \sar \Delta \sep \Sigma \sar A, \Pi \sep \lnsh$}
\RightLabel{$\nextri$}
\UnaryInfC{$\lnsg \sep \Gamma \sar \X A, \Delta \sep \Sigma \sar \Pi \sep \lnsh$}
\DisplayProof

&

\AxiomC{$\lnsg \sep \Gamma \sar \Delta \sep \sar A$}
\RightLabel{$\nextrii$}
\UnaryInfC{$\lnsg \sep \Gamma \sar \X A, \Delta$}
\DisplayProof

&

\AxiomC{$\lnsg \sep \Gamma^\annou, A \until B^\annof \sar \Delta \sep \lnsh$}
\RightLabel{$\annof$}
\UnaryInfC{$\lnsg \sep \Gamma^\annou, A \until B^\annou \sar \Delta \sep \lnsh$}
\DisplayProof
\end{tabular}
\end{center}

\begin{center}
\begin{tabular}{c}
\small
\AxiomC{$\lnsg \sep \Gamma, B \sar \Delta \sep \Sigma \sar \Pi \sep \lnsh$}
\AxiomC{$\lnsg \sep \Gamma, A \sar \Delta \sep \Sigma, A \until B^\an  \sar \Pi \sep \lnsh$}
\RightLabel{$\untilli$}
\BinaryInfC{$\lnsg \sep \Gamma, A \until B^\an \sar \Delta \sep \Sigma \sar \Pi \sep \lnsh$}
\DisplayProof
\end{tabular}
\end{center}

\begin{center}
\begin{tabular}{c}
\AxiomC{$\lnsg \sep \Gamma \sar A, B, \Delta \sep \Sigma \sar \Pi \sep \lnsh$}
\AxiomC{$\lnsg \sep \Gamma \sar B, \Delta \sep \Sigma \sar A \until B, \Pi \sep \lnsh$}
\RightLabel{$\untilri$}
\BinaryInfC{$\lnsg \sep \Gamma \sar A \until B, \Delta \sep \Sigma \sar \Pi \sep \lnsh$}
\DisplayProof
\end{tabular}
\end{center}


\begin{center}
\begin{tabular}{c c}
\AxiomC{$\lnsg \sep \Gamma, B \sar \Delta$}
\AxiomC{$\lnsg \sep \Gamma, A \sar \Delta \sep A \until B^\an \sar $}
\RightLabel{$\untillii$}
\BinaryInfC{$\lnsg \sep \Gamma, A \until B^\an \sar \Delta$}
\DisplayProof

&

\AxiomC{$\lnsg \sep \Gamma \sar A, B, \Delta$}
\AxiomC{$\lnsg \sep \Gamma \sar B, \Delta \sep \sar A \until B$}
\RightLabel{$\untilrii$}
\BinaryInfC{$\lnsg \sep \Gamma \sar A \until B, \Delta$}
\DisplayProof
\end{tabular}
\end{center}

\noindent\hrule
\caption{Inference Rules $\calc$.\label{fig:calc}}
\end{figure*}

A \emph{linear nested sequent (LNS)} is an expression of the form $\lnsg := \Gamma_{0} \sar \Delta_{0} \sep \cdots \sep \Gamma_{n} \sar \Delta_{n}$ such that (1) $\Gamma_{i} \sar \Delta_{i}$ is a Gentzen sequent for $0 \leq i \leq n$ called the \emph{i-component} of $\lnsg$ (or simply \emph{component} of $\lnsg$ if the index $i$ is not of importance), and (2) 
at most one component of $\lnsg$ is focused. LNSs are particularly natural for LTL, where every state in a state sequence has a unique successor. Intuitively, an LNS may thus be regarded as a finite `snapshot' of a state sequence, with each Gentzen sequent corresponding to one state and its right neighbor sequent corresponding to its successor state. We use $\lnsg$, $\lnsh$, $\lnsk$ and so on to denote LNSs. For an LNS $\lnsg = \Gamma_{0} \sar \Delta_{0} \sep \cdots \sep \Gamma_{n} \sar \Delta_{n}$ we refer to the $(n{-}1)$-component as the \emph{penultimate component}, to the $n$-component as the \emph{end component}, and to any $k$-component such that $k < n$ as an \emph{interior component}, that is, an interior component is any component that is \emph{not} the end component. If $\lnsg$ contains a formula in focus, then we call $\lnsg$ \emph{focused} and \emph{unfocused} otherwise. Given a focused LNS $\lnsg$, we write $\lnsg^\annou$ to denote the LNS obtained by changing the annotation of the focused formula to $\annou$.

We define the multiset of components of a linear nested sequent $\lnsg = \Gamma_{0} \sar \Delta_{0} \sep \cdots \sep \Gamma_{n} \sar \Delta_{n}$ to be $\comp{\lnsg} := \set{\Gamma_{i} \sar \Delta_{i} \mid 0 \leq i \leq n}$. The \emph{length} of an LNS $\lnsg$ is the number of components it contains, i.e., $\size{\lnsg} := |\comp{\lnsg}|$. For an LNS $\lnsg$, we let $\lnsg(i)$ be the $i$-component of $\lnsg$, if $0 \leq i < \size{\lnsg}$, and we let $\last{\lnsg} := \lnsg(\size{\lnsg}-1)$ denote the end component of $\lnsg$. For a Gentzen sequent $\Gamma \sar \Delta$, let $\sufo{\Gamma \sar \Delta} := \sufo{\Gamma^{-} \cup \Delta^{-}}$. For an LNS $\lnsg = \Gamma_0 \sar \Delta_0 \sep \cdots \sep \Gamma_n \sar \Delta_n$, we define $\sufo{\lnsg}$ as follows: $\sufo{\lnsg} := \bigcup_{0 \leq i \leq n} \sufo{\Gamma_i^{-} \cup \Delta_i^{-}}.$ We can interpret LNSs by means of their \emph{formula interpretation}, defined below:
$$
\fint(\Gamma \sar \Delta) := \bigwedge \Gamma^- \imp \bigvee \Delta^-
\quad
\fint(\Gamma \sar \Delta \sep \lnsg) := \bigwedge \Gamma^- \imp (\bigvee \Delta^- \lor \X \fint(\lnsg))
$$

\begin{definition}\label{def: satisfaction LNS} Let $\lnsg$ be an LNS, $0 \leq i < \size{\lnsg}$, and $\Bar{\sigma}$ be a state sequence. We define $\model, i \sat \lnsg(i)$ \iffi $\model, i \sat \bigwedge \Gamma^- \to \bigvee \Delta^-$ with $\lnsg(i) = \Gamma \sar \Delta$. We define $\model \sat \lnsg$ \iffi $\model \sat \fint(\lnsg)$, i.e., there exists $0 \leq i < \size{\lnsg}$ such that $\model, i \sat \lnsg(i)$.\footnote{Note that the annotations do not convey any semantic meaning.} An LNS $\lnsg$ is \emph{valid} \iffi for each state sequence $\Bar{\sigma}$, $\model \sat \lnsg$; otherwise, $\lnsg$ is \emph{invalid}.
\end{definition}

We define $\calc$ to be the set of all rules depicted in \fig~\ref{fig:calc}. The $\id$ and $\botl$ rules are \emph{initial rules} and we call the conclusion of an initial rule an \emph{initial sequent}. The \emph{focusing rules} $\annou$ and $\annof$, read bottom-up, unfocus and focus an until-formula, respectively. Note the condition that all side formulae in the antecedent of the displayed component are unfocused both in the conclusion and the premise. The remaining rules are \emph{logical rules} that introduce complex formulae into either the antecedent or consequent of a component of an LNS. As discussed in more detail below, the $\untilli$ and $\untillii$ rules utilize our focusing mechanism to ensure the soundness of non-wellfounded and cyclic proofs. Moreover, we remark that the $\untilli$, $\untillii$, $\untilri$, and $\untilrii$ rules rewrite an until-formula $A \until B$ into its equivalent unfolding $B \vee (A \wedge \X(A \until B))$. 
The rules $\untilli$ and $\untilri$ operate on until-formulae in an interior component: in the right premise, $A\until B$ is pushed into the next component. The rules $\untillii$ and $\untilrii$ work similarly but operate on until-formulae in the end component, thereby introducing a new component in the right premise.

The \emph{principal formula} of a logical rule in $\calc$ is the formula explicitly displayed in the conclusion, and we define the \emph{auxiliary formulae} to be those explicitly displayed in the premises. For example, $A \until B$ is principal in $\untilri$ and $A$, $B$, and $A \until B$ are auxiliary. All other formulae are called \emph{side formulae}. The \emph{principal component} (\emph{auxiliary components}) is (are) the component(s) where the principal (auxiliary, resp.) formula(e) occur.

The condition that LNSs contain at most one formula in focus imposes certain restrictions on rule instances, as illustrated by the following two lemmas.

\begin{lemma}\label{lem: properties focus rules}
The following hold.
    \begin{enumerate}
        \item[$(1)$] The premise of an instance of $\annou$ does not contain a formula in focus.
        \item[$(2)$] The conclusion of an instance of $\annof$ does not contain a formula in focus.
    \end{enumerate}
\end{lemma}
\begin{proof}
    For (1) consider an instance of $\annou$
    \begin{prooftree}
        \AxiomC{$\lnsg \sep \Gamma^\annou, A \until B^\annou \vdash \Delta \sep \lnsh$}
        \RightLabel{$\annou$}
        \UnaryInfC{$\lnsg \sep \Gamma^\annou, A \until B^\annof \vdash \Delta \sep \lnsh$}
    \end{prooftree}
    and suppose toward a contradiction that the premise contains a formula $C^\annof$ in focus. By the shape of the premise, $C^\annof$ must belong to a component in $\lnsg$ or in $\lnsh$. Therefore the conclusion contains two formulae in focus and is thus not a well-formed LNS, which is a contradiction. The case for (2) is similar.
\end{proof}

\begin{lemma}\label{lem: property until rule}
In every instance of $\untilli$ and $\untillii$ with the principal formula in focus, the left premise does not contain a formula in focus.
\end{lemma}
\begin{proof}
   For $\untilli$ consider an instance of $\untilli$ with the principal formula in focus in component $i$, as shown below.
    \begin{prooftree}
        \AxiomC{$\lnsg \sep \Gamma, B \sar \Delta \sep \Sigma \sar \Pi \sep \lnsh$}
        \AxiomC{$\lnsg \sep \Gamma, A \sar \Delta \sep \Sigma, A \until B^\annof \sar \Pi \sep \lnsh$}
        \RightLabel{$\untilli$}
        \BinaryInfC{$\lnsg \sep \Gamma, A \until B^\annof \sar \Delta \sep \Sigma \sar \Pi \sep \lnsh$}
    \end{prooftree}
    Suppose toward a contradiction that the left premise contains a formula $C^\annof$ in focus. By definition of $\untilli$, the auxiliary formula $B$ is unfocused, therefore $C^\annof$ occurs either in $\lnsg$ or $\lnsh$, or belongs to $\Gamma$ or $\Sigma$. In either case, $C^\annof$ also occurs in the right premise. The right premise contains the displayed focused occurrence in component \(i+1\). Hence any focused side formula would have to coincide with that occurrence; otherwise, the right premise would have two focused formula occurrences. Thus a focused side formula could only be \(AUB^{\mathsf f}\in\Sigma\). But then the conclusion would contain focused formula occurrences in both components \(i\) and \(i+1\), which is impossible. Therefore there are no focused side occurrences, and the left premise is unfocused. The case for $\untillii$ is similar.
\end{proof}

Since we are working with sequents based on sets of formulae (where contraction is implicit), rule instances are either preserving or succinct. A rule instance is \emph{preserving} if the principal formula occurs as a side formula on the same side of the principal component in the premise, and \emph{succinct} otherwise. For example, the following shows a preserving instance of $\impr$ on the left and a succinct instance on the right:
\begin{center}
\begin{tabular}{c c}
\AxiomC{$\lnsg \sep  A \sar  B, A \imp B \sep \lnsh$}
\RightLabel{$\impr$}
\UnaryInfC{$ \lnsg \sep \sar A \imp B \sep \lnsh$}
\DisplayProof

&

\AxiomC{$\lnsg \sep  A \sar  B \sep \lnsh$}
\RightLabel{$\impr$}
\UnaryInfC{$\lnsg \sep  \sar A \imp B \sep \lnsh$}
\DisplayProof
\end{tabular}
\end{center}

\subsection{Non-Wellfounded Proofs} 

We define a \emph{derivation} $\der$ of an LNS $\lnsg$ in $\calc$ to be a (potentially infinite) tree whose nodes are labelled by LNSs such that (1) the root is labelled by $\lnsg$ and (2) every parent node is labelled by the conclusion of a rule of $\calc$ with its children labelled by the corresponding premises. We will usually identify the node with the LNS that labels it, and treat distinct occurrences of the same LNS as distinct nodes. If there is any danger of confusion, we will make the distinction between nodes and their labels explicit. An \emph{index set} is a set $\idxset \subseteq \mathbb{N}$ such that (1) $0 \in \idxset$ and (2) if $n+1 \in \idxset$, then $n \in \idxset$, i.e., an index set is an initial segment of the natural numbers, or the whole set of naturals. We use index sets to define certain kinds of (potentially infinite) sequences below. 

A \emph{path} in a derivation $\der$ is a (potentially infinite) sequence $\lnspath = \set{\lnsg_{i}}_{i \in \idxset}$ of LNSs (where $\idxset$ is an index set) such that for each $i \in \idxset$ if $i+1 \in \idxset$, then the LNS $\lnsg_{i}$ is the parent of $\lnsg_{i+1}$. A path $\lnspath = \set{\lnsg_{i}}_{i \in \idxset}$ is \emph{finite} if $\idxset \not = \mathbb{N}$ and \emph{infinite} otherwise. Moreover, $\lnspath = \set{\lnsg_{i}}_{i \in \idxset}$ is \emph{maximal} if $\lnspath$ is finite and $\lnsg_{\lvert \idxset \rvert -1}$ is a leaf or if $\lnspath$ is infinite. A \emph{branch} in a derivation $\der$ is a maximal path $\branch = \set{\lnsg_{i}}_{i \in \idxset}$ of LNSs such that $\lnsg_{0}$ is the root. The \emph{height} of a derivation is defined in the usual way as the maximal length of a branch in the derivation, which may be infinite. A path $\lnspath := \set{\lnsg_{i}}_{i \in \idxset}$ has a \emph{continuous focus trace} \iffi for every $i \in \idxset$, $\lnsg_{i}$ contains a formula in focus. The following lemma follows directly from this definition and Lemma~\ref{lem: properties focus rules}:

\begin{lemma}
If a path $\lnspath$ has a continuous focus trace, then neither $\annou$ nor $\annof$ are applied along the path.
\end{lemma}

Observe that on a continuous focus trace, the underlying until-formula remains the same. An infinite path $\lnspath$ in a derivation $\der$ is called \emph{progressing} if (1) $\lnspath$ has a continuous focus trace and (2) $\lnspath$ passes infinitely often through the right premise of an instance of a rule $\mathsf{r} \in \{\untilli, \untillii\}$ in which the formula in focus is principal. An infinite branch $\branch$ in $\der$ is called \emph{progressing} if it has a suffix which is a progressing path.

\begin{definition}
A \emph{non-wellfounded proof} (or, \emph{$\calci$-proof}) of an LNS $\lnsg$ is a derivation $\der$ of $\lnsg$ such that (1) every leaf of $\der$ is an initial sequent and (2) every infinite branch is progressing.
\end{definition}
We let $\calci$ be the non-wellfounded LNS calculus obtained by letting the set of provable sequents be determined by non-wellfounded proofs and write $\calci \proves \lnsg$ if $\lnsg$ has a non-wellfounded proof.

\begin{example}\label{ex: nwf proof}
    An example of a non-wellfounded proof $\pi$ of $p \until q^\annof \sar p \until q$ is depicted in Figure~\ref{fig:example-nwf-proof}. The proof consists of one infinite branch (the right-most one indicated by the three dots) and infinitely many finite branches (which eventually branch to the left). Every finite branch ends in an initial sequent, while the infinite branch contains a continuous focus trace which is progressing.
\end{example}

\begin{figure}
\begin{center}
\begin{tabular}{c c @{\hspace{0.2cm}} c c}
$\prf_{1} \ :=$
&
\small
\AxiomC{}
\RightLabel{$\id$}
\UnaryInfC{$q \sar p, q$}
\AxiomC{}
\RightLabel{$\id$}
\UnaryInfC{$q \sar q \sep \sar p \until q$}
\RightLabel{$\untilrii$}
\BinaryInfC{$q \sar p \until q$}
\DisplayProof
&
$\prf_{2} \ :=$
&
\small
\AxiomC{}
\RightLabel{$\id$}
\UnaryInfC{$p \sar q \sep q \sar p, q$}
\AxiomC{}
\RightLabel{$\id$}
\UnaryInfC{$p \sar q \sep q \sar q \sep \sar p \until q$}
\RightLabel{$\untilrii$}
\BinaryInfC{$p \sar q \sep q \sar p \until q$}
\DisplayProof
\end{tabular}
\end{center}

\begin{center}
\small
{\AxiomC{$\prf_{1}$ }

\AxiomC{}
\RightLabel{$\id$}
\UnaryInfC{$p \sar p, q \sep p \until q^\annof \sar$}

\AxiomC{$\pi_2$}

\AxiomC{}
\RightLabel{$\id$}
\UnaryInfC{$p \sar q \sep p \sar p,q \sep p \until q^\annof \sar$}

\AxiomC{$.$}
\noLine
\UnaryInfC{$.$}
\noLine
\UnaryInfC{$.$}
\noLine
\UnaryInfC{$p \sar q \sep p \sar q \sep p \until q^\annof \sar p \until q$}
\RightLabel{$\untilri$}
\BinaryInfC{$p \sar q \sep p \sar p \until q \sep  p \until q^\annof \sar$}
\RightLabel{$\untillii$}
\BinaryInfC{$p \sar q \sep p \until q^\annof \sar p \until q$}
\RightLabel{$\untilri$}
\BinaryInfC{$p \sar p \until q \sep p \until q^\annof \sar$}
\RightLabel{$\untillii$}
\BinaryInfC{$p \until q^\annof \sar p \until q$}
\DisplayProof}
\end{center}

\caption{Example of a non-wellfounded proof.\label{fig:example-nwf-proof}}
\end{figure}

\subsection{Cyclic Proofs} 

Let $\lnsg$ be an LNS with $\size{\lnsg} = n+1$ and $\lnsg(i) = \Gamma_{i} \sar \Delta_{i}$. For $0 \leq i \leq n$, we define $\lnsg$ to be \emph{$i$-saturated} \iffi $\set{A \mid A \in (\Gamma_{i} \cup \Delta_{i})^-} \cap \set{\X A, A \imp B, A \until B \mid A, B \in \lang} = \emptyset$, i.e., its $i$-component does not contain any complex formulae. We call $\lnsg$ \emph{saturated} \iffi $\lnsg$ is not an initial sequent and for each $0 \leq i < n$, $\lnsg$ is $i$-saturated. Intuitively, saturated sequents serve as checkpoints: since their interior components are free of complex formulae (meaning, no rules are bottom-up applicable to these components), two saturated sequents with matching end components represent identical `states' in a proof, making them natural candidates for the endpoints of a cycle. This serves as the motivation for our notion of cyclic proof given below.

\begin{definition}\label{dfn:cyclic-proof}
A tuple $\tuple{\der,L,c}$ is a \emph{cyclic proof} (or, \emph{$\calccyc$-proof}) of an LNS $\lnsg$ \iffi it satisfies: 
\begin{itemize}

\item[$(1)$] $\der$ is a finite derivation of $\lnsg$;

\item[$(2)$] $L$ is a set containing only leaves of $\der$, which we call \emph{buds};

\item[$(3)$] $c$ is a function mapping each $\lnsg \in L$ to an LNS $c(\lnsg)$ strictly below $\lnsg$, called the \emph{companion} of $\lnsg$, such that $\size{c(\lnsg)} < \size{\lnsg}$, the path from $c(\lnsg)$ to $\lnsg$ has a continuous focus trace and goes through the right premise of a rule $\ru \in \set{\untilli,\untillii}$ with the principal formula in focus, both $c(\lnsg)$ and $\lnsg$ are saturated, and $\last{c(\lnsg)} = \last{\lnsg}$;\footnote{By ``\emph{strictly below $\lnsg$},'' we mean that $c(\lnsg)$ occurs on the path from the root of $\der$ to the leaf $\lnsg$ and is distinct from $\lnsg$.}


\item[$(4)$] every other leaf in $\der$ that is not in $L$ is an initial sequent.

\end{itemize}
\end{definition}
We let $\calccyc$ be the cyclic LNS calculus obtained by letting the set of provable sequents be determined by cyclic proofs and write $\calccyc \proves \lnsg$ if $\lnsg$ has a cyclic proof.

\begin{example} An example of a cyclic proof $(\der, L, c)$ is displayed in \fig~\ref{fig:example-cyc-proof} where $L = \set{\lnsg}$ and $c(\lnsg) = \lnsh$ such that $\lnsg := p \sar q \sep p \until q^\annof \sar p \until q$ and $\lnsh := p \until q^\annof \sar p \until q$. It is straightforward to verify that $c$ defines a correct cycle since $\size{\lnsh} < \size{\lnsg}$, the branch from $\lnsh$ to $\lnsg$ has a continuous focus trace and goes through the right premise of $\untillii$ with the principal formula in focus, both $\lnsg$ and $\lnsh$ are saturated, and $\last{\lnsh} = \last{\lnsg}$.
\end{example}

\begin{figure}
\begin{center}
\begin{tabular}{c c}
$\prf_{1} \ :=$
&
\AxiomC{}
\RightLabel{$\id$}
\UnaryInfC{$q \sar p, q$}
\AxiomC{}
\RightLabel{$\id$}
\UnaryInfC{$q \sar q \sep \sar p \until q$}
\RightLabel{$\untilrii$}
\BinaryInfC{$q \sar p \until q$}
\DisplayProof
\end{tabular}
\end{center}

\begin{center}
\begin{tikzpicture}
\node[] (a) [] {\AxiomC{$\prf_{1}$}

\AxiomC{}
\RightLabel{$\id$}
\UnaryInfC{$p \sar p, q \sep p \until q^\annof \sar$}

\AxiomC{$p \sar q \sep p \until q^\annof \sar p \until q$}
\RightLabel{$\untilri$}
\BinaryInfC{$p \sar p \until q \sep p \until q^\annof \sar$}
\RightLabel{$\untillii$}
\BinaryInfC{$p \until q^\annof \sar p \until q$}
\DisplayProof};


\draw[->,rounded corners=.25cm,dashed]
  (2.8,.75) -- (2.8,1) -- (6.2,1) -- node[right,midway]{$c$} 
  (6.2,-.65) -- (-.2,-.65);

\end{tikzpicture}
\end{center}

\caption{Example of a cyclic proof.\label{fig:example-cyc-proof}}
\end{figure}

\section{Soundness of Non-Wellfounded Proofs}\label{sec: soundness}

In this section we prove that $\calci$ is sound. We proceed via proof by infinite descent. We assume, toward a contradiction, that some invalid LNS $\lnsg$ has a non-wellfounded proof $\pi$. Since $\lnsg$ is invalid, there exists a state sequence $\model$ such that $\model \not \sat \lnsg$. We then show how to obtain a branch $\branch =\{\lnsg_i\}_{i\in \mathbb{N}}$ of $\pi$ such that $\model \not\sat \lnsg_i$ for each $i \in \mathbb{N}$, by showing that all rules of $\calci$ are locally sound. Note that $\branch$ must be an infinite branch of $\pi$, since initial sequents are always valid. Simultaneously, we assign a well-founded measure in the form of natural numbers to formulae in focus occurring in LNSs along $\branch$ and show that whenever a formula $A \until B^\annof$ is principal in $\untilli$ or $\untillii$ and $\branch$ passes through the right premise, the measure assigned to $A \until B^\annof$ strictly decreases, and weakly decreases otherwise. Finally, since $\prf$ is a proof, $\branch$ must contain a progressing suffix, meaning that there is a continuous focus trace and the suffix passes infinitely often through the right premise of instances of $\untilli$ and $\untillii$ with the formula in focus principal. Thus, we obtain an infinite decreasing sequence of natural numbers which strictly decreases infinitely often, contradicting the well-foundedness of 
$\mathbb{N}$. This gives the desired contradiction, implying the soundness of $\calci$.

Let $\lnsg$ be an LNS and recall that only formulae of the form $A \until B$ can occur in focus in an LNS $\lnsg$ and such formulae always occur in the antecedent of a component. A \emph{progress measure} for $\lnsg$ is a partial function which assigns to the 
formula $A \until B^{\annof}$ in focus in $\lnsg$ (if such a formula exists) a natural number; if no formula is in focus, the progress measure is undefined.

\begin{theorem}[Soundness of $\calci$]\label{thm:soundness}
    If $\calci \proves \lnsg$, then $\lnsg$ is valid.
\end{theorem}

\begin{proof} Let $\pi$ be a non-wellfounded proof of $\lnsg$ and suppose, for contradiction, that $\lnsg$ is invalid. Then, there exists a state sequence $\model$ such that $\model \not \sat \lnsg$. We inductively define an infinite branch of LNSs $\branch=\set{\lnsg_i}_{i \in \mathbb{N}}$ through $\pi$ and an infinite sequence of progress measures $\set{\mu_i}_{i \in \mathbb{N}}$ such that for all $i \in \mathbb{N}$ the following hold:
\begin{enumerate}

\item[$(1)$] $\model \not \sat \lnsg_i$;

\item[$(2)$] $\mu_i$ is a progress measure for $\lnsg_i$;

\item[$(3)$] if a formula $A \until B^{\annof}$ in focus occurs in the $j$-component of $\lnsg_i$, then the following hold:

\begin{enumerate}

\item $\mu_i(A \until B^\annof) = min\{k \in \mathbb{N} \mid \Bar{\sigma}, j+k \models B \text{ and } \forall k' < k: \Bar{\sigma}, j+k' \models A\}$;

\item for $i > 0$ if $A \until B^{\annof}$ is in focus in $\lnsg_{i-1}$, then $\mu_i(A \until B^\annof) \leq \mu_{i-1}(A \until B^\annof)$;

\item for $i > 0$ if $A \until B^{\annof}$ is in focus in $\lnsg_{i-1}$ and $\lnsg_{i-1}$ is the conclusion of an instance of $\mathsf{r} \in \{\untilli, \untillii\}$ with $A \until B^{\annof}$ principal, then $\mu_i(A \until B^\annof) < \mu_{i-1}(A \until B^\annof)$.\footnote{Observe that $\lnsg_i$ must be the right premise of $\mathsf{r}$ due to Lemma~\ref{lem: property until rule}.}

\end{enumerate}
\end{enumerate}
We now define $\set{\lnsg_i}_{i \in \mathbb{N}}$ and $\set{\mu_i}_{i \in \mathbb{N}}$. First, let $\lnsg_0 := \lnsg$. By assumption $\model \not \sat \lnsg$, so (1) is satisfied. Moreover, by Definition~\ref{def: satisfaction LNS}, this implies that if a formula $A \until B^\annof$ occurs in the antecedent of the $j$-th component of $\lnsg_0$ (for $0 \leq j < \size{\lnsg_0}$), $\model, j \sat A \until B$ and so there exists a natural number $k$ such that (i) $\model, j+k \sat B$ and (ii) $\model, j+k' \sat A$ for all $0 \leq k' < k$. If a formula in focus occurs in $\lnsg$, then let $\mu_0$ be the partial function that assigns the formula $A \until B^\annof$ the least natural number $k$ such that (i) and (ii) hold. Otherwise let $\mu_0$ be undefined. Observe that conditions (1)--(3) are satisfied. 

Suppose we have defined $\lnsg_i$ and $\mu_i$ and they satisfy conditions (1)--(3). We define $\lnsg_{i+1}$ and $\mu_{i+1}$ by using a case distinction based on the rule $\mathsf{r}$ with conclusion $\lnsg_i$ in $\pi$. Note that $\mathsf{r}$ cannot be an initial rule, since $\model \not \sat \lnsg_i$ by (1). We show the cases for ${\to}\mathsf{R}$, $\annou$, $\annof$, $\X\mathsf{L}_2$, $\untilli$ and $\untilrii$. Each remaining case is similar to one of these cases.

\case{Case for ${\to}\mathsf{R}$.} Suppose $\lnsg_i$ is the conclusion of an instance of ${\to}\mathsf{R}$, as shown below.
\begin{prooftree}
    \AxiomC{$\overbrace{\lnsg \sep \Gamma, A \vdash B, \Delta \sep \lnsh}^{\lnsf}$}
    \RightLabel{${\to}\mathsf{R}$}
    \UnaryInfC{$\underbrace{\lnsg \sep \Gamma \vdash A \to B, \Delta \sep \lnsh}_{\lnsg_i}$}
\end{prooftree}
We assume that the principal formula $A \to B$ occurs in component $j$. By condition (1), $\Bar{\sigma} \not \models \lnsg_i$, which implies that $\Bar{\sigma}, j \models \bigwedge \Gamma^-$, $\Bar{\sigma}, j \not \models \bigvee \Delta^-$, and $\Bar{\sigma}, j \not \models A \to B$. The latter implies that $\Bar{\sigma}, j \models A$ and $\Bar{\sigma}, j \not \models B$. Hence $\Bar{\sigma} \not \models \lnsf$. Let $\lnsg_{i+1} := \lnsf$ and note that (1) is satisfied. Note that $\lnsg_{i+1}$ contains a formula occurrence in focus \iffi $\lnsg_i$ contains the same formula occurrence in focus. Therefore we define $\mu_{i+1} := \mu_i$. Note that (2) and (3) are satisfied. 

\case{Case for $\annou$.} Suppose $\lnsg_i$ is the conclusion of an instance of $\annou$, as shown below.
\begin{prooftree}
    \AxiomC{$\overbrace{\lnsg \sep \Gamma^\annou, A \until B^\annou \vdash \Delta \sep \lnsh}^{\lnsf}$}
    \RightLabel{$\annou$}
    \UnaryInfC{$\underbrace{\lnsg \sep \Gamma^\annou, A \until B^\annof \vdash \Delta \sep \lnsh}_{\lnsg_i}$}
\end{prooftree}
We assume that the formula $A \until B^\annof$ occurs in component $j$. By condition (1) we have $\Bar{\sigma} \not \models \lnsg_i$. Since annotations do not convey any semantic meaning, we have $\Bar{\sigma} \not \models \lnsf$. Therefore let $\lnsg_{i+1} := \lnsf$ and note that condition (1) is satisfied. Moreover, by Lemma~\ref{lem: properties focus rules}, $\lnsg_{i+1}$ does not contain a formula in focus, implying that condition (3) is vacuously satisfied. We define $\mu_{i+1}$ to be undefined. Note that (2) is then satisfied as well.

\case{Case for $\annof$.} Suppose $\lnsg_i$ is the conclusion of an instance of $\annof$, as shown below.
\begin{prooftree}
    \AxiomC{$\overbrace{\lnsg \sep \Gamma^\annou, A \until B^\annof \sar \Delta \sep \lnsh}^{\lnsf}$}
    \RightLabel{$\annof$}
    \UnaryInfC{$\underbrace{\lnsg \sep \Gamma^\annou, A \until B^\annou \sar \Delta \sep \lnsh}_{\lnsg_i}$}
\end{prooftree}
We assume that the formula $A \until B^\annof$ occurs in component $j$. By condition (1) we have that $\Bar{\sigma} \not \models \lnsg_i$. Since annotations do not convey any semantic meaning, we obtain $\Bar{\sigma} \not \models \lnsf$. Thus let $\lnsg_{i+1} := \lnsf$ and note that condition (1) is satisfied. Moreover, $\Bar{\sigma}, j \models A \until B$, hence there exists a least natural number $k$ with $\Bar{\sigma}, j+k \models B$ and $\Bar{\sigma}, j + k' \models A$ for all $k' < k$. Set $\mu_{i+1}(A \until B^\annof) := k$ and note that condition (2) and (3)(a) are satisfied. Finally, by Lemma~\ref{lem: properties focus rules}--(2) the conclusion $\lnsg_i$ does not have a formula in focus. Hence conditions (3)(b) and (c) vacuously hold.

\case{Case for $\nextlii$.} Suppose $\lnsg_i$ is the conclusion of an instance of $\nextlii$, as shown below.
\begin{prooftree}
    \AxiomC{$\overbrace{\lnsg \sep \Gamma \sar \Delta \sep A \sar }^{\lnsf}$}
    \RightLabel{$\nextlii$}
    \UnaryInfC{$\underbrace{\lnsg \sep \Gamma, \X A \sar \Delta}_{\lnsg_i}$}
\end{prooftree}
By condition (1) we have that $\Bar{\sigma} \not \models \lnsg_i$, which implies that $\Bar{\sigma}, j \models \bigwedge \Gamma^-$, $\Bar{\sigma}, j \models \X A$ and $\Bar{\sigma}, j \not \models \bigvee \Delta^-$, where $j = \size{\lnsg_i} -1$. Therefore $\Bar{\sigma}, j+1 \models A$ and so $\Bar{\sigma} \not \models \lnsf$. Let $\lnsg_{i+1} := \lnsf$ and note that condition (1) is satisfied. Moreover $\lnsg_{i+1}$ contains a formula occurrence in focus \iffi $\lnsg_i$ contains the same formula occurrence in focus. Therefore we let $\mu_{i+1} := \mu_i$. Note that conditions (2) and (3) are satisfied.

\case{Case for $\untilli$.} Suppose $\lnsg_i$ is the conclusion of an instance of $\untilli$, as shown below.
\begin{center}
\small
\begin{tabular}{c}
\AxiomC{$\overbrace{\lnsg \sep \Gamma, B \sar \Delta \sep \Sigma \sar \Pi \sep \lnsh}^{\lnsf_1}$}
\AxiomC{$\overbrace{\lnsg \sep \Gamma, A \sar \Delta \sep \Sigma, A \until B^\an  \sar \Pi \sep \lnsh}^{\lnsf_2}$}
\RightLabel{$\untilli$}
\BinaryInfC{$\underbrace{\lnsg \sep \Gamma, A \until B^\an \sar \Delta \sep \Sigma \sar \Pi \sep \lnsh}_{\lnsg_i}$}
\DisplayProof
\end{tabular}
\end{center}
We assume that the principal formula $A \until B^\an$ occurs in component $j$ and we distinguish two cases. First, suppose that $\an = \annou$. By condition (1) we have that $\Bar{\sigma} \not \models \lnsg_i$, which implies that $\Bar{\sigma}, j \models \bigwedge \Gamma^-$, $\Bar{\sigma}, j \models A \until B$ and $\Bar{\sigma}, j \not \models \bigvee \Delta$. Hence there exists a natural number $k$ such that $\Bar{\sigma}, j+k \models B$ and $\Bar{\sigma}, j + k' \models A$ for all $0 \leq k' < k$. If $k = 0$, then $\Bar{\sigma}, j \models B$ and so $\Bar{\sigma} \not \models \lnsf_1$. In this case let $\lnsg_{i+1} := \lnsf_1$. If $k > 0$, then $\Bar{\sigma}, j \models A$ and $\Bar{\sigma}, j+1 \models A \until B$, implying that $\Bar{\sigma} \not \models \lnsf_2$. In that case let $\lnsg_{i+1} := \lnsf_2$. Observe that in both cases (1) holds. There are now two sub-cases to consider: either (i) $\lnsg_{i}$ contains a formula in focus or (ii) it does not. In case (i), the formula in focus is a side formula of the rule instance and 
therefore occurs in the same component in both $\lnsg_i$ and $\lnsg_{i+1}$; hence, we take $\mu_{i+1} := \mu_i$ and observe that (2) and (3) hold, where (3)(b) is satisfied with equality. In case (ii), we take $\mu_{i+1}$ to be undefined and note that (2) and (3) hold vacuously.

Second, suppose that $\an = \annof$. We make a case distinction on if (i) $\mu_i(A \until B^\annof) = 0$ or (ii) $\mu_i(A \until B^\annof) > 0$.

(i) If $\mu_i(A \until B^\annof) = 0$, then $\model, j \sat B$ by (3)(a). Note that then $\model \not\sat \lnsf_1$ and so let $\lnsg_{i+1} := \lnsf_1$. Note that condition (1) is satisfied. By Lemma~\ref{lem: property until rule} no formula is in focus in $\lnsg_{i+1}$. We therefore take $\mu_{i+1}$ to be undefined, which implies that condition (2) holds. Finally, note that condition (3) is vacuously satisfied.

(ii) If $\mu_i(A \until B^\annof) > 0$, then observe that $\model, j \sat A$ and $\model, j+1 \sat A \until B$ by (3)(a). Therefore, $\model \not\sat \lnsf_2$ and so let $\lnsg_{i+1} := \lnsf_2$. Moreover, since $\mu_i(A \until B^\annof)$ is the least natural number $k$ such that $\model, j+k \sat B$ and $\model, j+k' \sat A$ for each $0 \leq k' < k$, we have that $k-1$ is the least natural number such that  $\model, (j+1)+(k-1) \sat B$ and $\model, (j+1)+k' \sat A$ for each $0 \leq k' < k-1$. Let $\mu_{i+1}(A \until B^\annof):= \mu_i(A \until B^\annof)-1$. Note that conditions (1)--(3) are satisfied.

\case{Case for $\untilrii$.} Suppose $\lnsg_i$ is the conclusion of an instance of $\untilrii$, as shown below.
\begin{prooftree}
    \AxiomC{$\overbrace{\lnsg \sep \Gamma \sar A, B, \Delta}^{\lnsf_1}$}
    \AxiomC{$\overbrace{\lnsg \sep \Gamma \sar B, \Delta \sep \sar A \until B}^{\lnsf_2}$}
    \RightLabel{$\untilrii$}
    \BinaryInfC{$\underbrace{\lnsg \sep \Gamma \sar A \until B, \Delta}_{\lnsg_i}$}
\end{prooftree}
By condition (1) we have that $\Bar{\sigma} \not \models \lnsg_i$, which implies that $\Bar{\sigma}, j \models \bigwedge \Gamma^-$, $\Bar{\sigma}, j \not \models \bigvee \Delta^-$ and $\Bar{\sigma}, j \not \models A \until B$, where $j = \size{\lnsg_i}-1$. Note that $\Bar{\sigma}, j \not \models A \until B$ implies that $\Bar{\sigma}, j \not \models B$. We distinguish two cases. (i) If $\Bar{\sigma}, j \not \models A$, let $\lnsg_{i+1} := \lnsf_1$. (ii) If $\Bar{\sigma}, j \models A$, then $\Bar{\sigma}, j+1 \not \models A \until B$, and so we let $\lnsg_{i+1} := \lnsf_2$. Note that in either case condition (1) is satisfied. Moreover $\lnsg_{i+1}$ contains a formula occurrence in focus \iffi the same formula occurrence is in focus in $\lnsg_i$ and so we let $\mu_{i+1} := \mu_i$. Note that conditions (2) and (3) are satisfied. This concludes the construction of $\branch = \set{\lnsg_i}_{i \in \mathbb{N}}$ and of $\set{\mu_i}_{i \in \mathbb{N}}$.

Since $\pi$ is a proof, the infinite branch $\branch =\set{\lnsg_i}_{i \in \mathbb{N}}$ must have a suffix $\lnspath= \set{\lnsg_i}_{i \geq j}$ that is a progressing path. Hence, $\lnspath$ has a continuous focus trace. Moreover, $\lnspath$ passes infinitely often through the right premise of a rule $\mathsf{r} \in \{\untilli, \untillii\}$ with the principal formula $A \until B^{\annof}$ in focus. Therefore, $\mu_i(A \until B^\annof)$ is defined for all $i \geq j$. Consider the infinite sequence of natural numbers $\set{\mu_i(A \until B^\annof)}_{i\geq j}$. By (3)(b) and (3)(c) we have that for each $i \geq j$, $\mu_{i+1}(A \until B^\annof) \leq \mu_i(A \until B^\annof)$. Since $\lnspath$ passes infinitely often through the right premise of a rule $\mathsf{r} \in \{\untilli, \untillii\}$ with the formula $A \until B^\annof$ principal, (3)(c) implies that there are infinitely many $i \geq j$ with  $\mu_{i+1}(A \until B^\annof) < \mu_i(A \until B^\annof)$. Therefore, $\set{\mu_i(A \until B^\annof)}_{i\geq j}$ is an infinite decreasing sequence of natural numbers that strictly decreases infinitely often. This contradicts that the natural numbers are well-founded. Hence, $\lnsg$ must be valid. 
\end{proof}

In the proof of Theorem~\ref{thm:soundness}, we have also established that all rules of $\calci$ are locally sound: if the premises are valid, then the conclusion is valid as well. Similarly, we can establish semantic invertibility: if the conclusion is valid, then all premises are valid. The proof of the following lemma follows immediately by the semantic clauses and the interpretation of LNSs.
\begin{lemma}[Semantic invertibility]\label{lem: semantic invertibility}
    Consider an instance of a rule $\ru \in \calci$ with conclusion $\lnsg$ and premises $\lnsh_i$ for $1 \leq i \leq n$. If $\lnsg$ is valid, then $\lnsh_i$ is valid for every $1 \leq i \leq n$.
\end{lemma}

\section{Completeness of Non-Wellfounded Proofs}\label{sec: completeness}

In this section we prove completeness of $\calci$ by a proof-search argument. Given an LNS $\lnsg$, we build a structure $\prf$, called a \emph{proof-search tree}, by a phased construction. The construction first exhaustively applies rules bottom-up to formulae in the interior components of LNSs until every open leaf is either initial or saturated. Initial leaves are closed, as are saturated leaves that are not the conclusion of any rule instance. For any other saturated leaf, we first apply a focus rule to one of the until-formulae in the end-component (if any exist). Since there are in general multiple formulae present that may be focused or unfocused, we introduce so-called \emph{branching focus rules} which have one premise for each possible application of $\annou$ or $\annof$. Then we apply a logical rule to one of the complex formulae in the end-component. If we apply a rule that adds a new component, then the premises are in general not saturated, and so the process can be repeated indefinitely. This results in a potentially infinite tree of LNSs called a proof-search tree for $\lnsg$. Since we employ branching focus rules, this proof-search tree is not in general a derivation in $\calci$. However, we can obtain derivations from proof-search trees by choosing a single premise for each branching focus rule instance and deleting the other branches. Derivations obtained in that way satisfy the \emph{saturation recurrence property (SRP)}:

\begin{definition} A derivation $\pi$ satisfies the \emph{saturation recurrence property (SRP)} if and only if, for any infinite branch $\branch := \{\lnsg_i\}_{i \in \mathbb{N}}$ in $\pi$, infinitely many $\lnsg_i$ are saturated.
\end{definition}
 
We then show that given a proof-search tree for an LNS $\lnsg$, we can either extract a derivation from it which is a non-wellfounded proof of $\lnsg$ that satisfies the SRP, or the proof-search tree contains a non-progressing branch which can be used to construct a counter-model for $\lnsg$. Since valid LNSs do not have counter-models, we obtain that every valid LNS has a non-wellfounded proof that satisfies the SRP. Thus we obtain completeness of $\calci$ with respect to SRP-proofs. This property will prove useful in the transformation of non-wellfounded into cyclic proofs in Section~\ref{sec:correspondence}. We begin by introducing the branching focus rules and proof-search trees.

\subsection{Proof-Search Trees}

We introduce two new inference rules that will be used in the completeness proof for $\calc$. These are taken to be \emph{disjunctive rules} (cf.\cite{Min00,LyoTah26}), meaning, the conclusion is provable whenever at least one premise is provable. By contrast, we refer to the remaining rules in $\calc$ as \emph{conjunctive rules}: the conclusion is provable whenever all premises are provable. Our proof-search tree will consist of both conjunctive and disjunctive rule applications, serving as a `macro-structure' from which we can extract either proofs or counter-models (cf.~\cite{LyoTah26,Lyo26}).

\begin{definition}
    Let $\lnsg = \Gamma_{0} \sar \Delta_{0} \sep \cdots \sep \Gamma_{n} \sar \Delta_{n}$ be unfocused and define
    \begin{equation*}
        U(\lnsg) := \set{(i, C) \mid 0 \leq i \leq n : C= A \until B \text{ for some } A,B \in \lang, C^\annou \in \Gamma_{i}}.
    \end{equation*}
    For $(i, C) \in U(\lnsg)$, let $\lnsg^{(i,C)}$ be obtained by replacing $C^\annou$ in the antecedent of component $i$ by $C^\annof$, leaving everything else unchanged.\smallskip
    
    The \emph{focused branching rule} $\fbr$ has exactly the following premises
    \begin{center}
\AxiomC{$\lnsg$}
\AxiomC{$\{\lnsg^{(i,C)}\}_{(i,C) \in U(\lnsg)}$}
\RightLabel{$\fbr$.}
\BinaryInfC{$\lnsg$}
\DisplayProof
\end{center}
When $U(\lnsg) = \emptyset$, this is a unary identity step. 
\end{definition}

In essence, this rule takes an unfocused LNS as the conclusion, and via a bottom-up application, focuses all formulae of the form $A \until B$ in antecedents of components in all possible ways (as well as copying the conclusion as a premise). 

\begin{definition}
    The \emph{unfocused branching rule} $\ubr$ is given by
    \begin{center}
\AxiomC{$\lnsg^\annou$}
\AxiomC{$\lnsg$}
\RightLabel{$\ubr$}
\BinaryInfC{$\lnsg$}
\DisplayProof
\end{center}
where $\lnsg$ is focused.
\end{definition}

This rule takes an LNS in focus, and via a bottom-up application, returns one premise with the formula unfocused and another premise identical to the conclusion. 

We define the calculus $\pcalc := (\calc \setminus \set{\annou,\annof}) \cup \set{\fbr,\ubr}$. A \emph{pseudo-derivation} is defined to be a derivation in $\pcalc$. 

\begin{definition}
The \emph{trivial rule} $\mathsf{triv}$ is defined by
\begin{center}
\AxiomC{$\lnsg$}
\RightLabel{$\mathsf{triv}$.}
\UnaryInfC{$\lnsg$}
\DisplayProof
\end{center}    
\end{definition}

Observe that if a trivial rule application occurs in a proof, then it can be `contracted' to a single copy of $\lnsg$, showing its admissibility in $\calc$.

Let us call $\nextlii$, $\nextrii$, $\untillii$ and $\untilrii$ \emph{expansion} rules and all other rules in $\pcalc$ \emph{length-preserving}. An LNS $\lnsg$ is called \emph{fully saturated} \iffi $\lnsg$ is not initial and for $0 \leq i < \size{\lnsg}$, $\lnsg$ is $i$-saturated, i.e., $\lnsg$ is not the conclusion of any instance of a rule in $\calci$. We formalize proof-search by means of \emph{proof-search trees} which represent a systematic search for a proof of a given LNS. We provide a phased construction of proof-search trees, and show that it produces a pseudo-derivation.

\begin{definition}\label{def:proof-search-tree} The \emph{phased construction} is given as follows. Every logical rule instance used in the construction is succinct.
\begin{itemize}
    \item[$\triangleright$] \textbf{Saturation Phase.} At the current node $\lnsh$:
    \begin{enumerate}
        \item[$(1)$] If $\lnsh$ is initial, close the node.
        \item[$(2)$] If $\lnsh$ is not saturated, choose a complex formula in an interior component and apply the corresponding logical rule bottom-up. Restart the saturation phase at every premise.
        \item[$(3)$] Otherwise $\lnsh$ is saturated.
        \begin{enumerate}
            \item If $\lnsh$ is fully saturated, close the node.
            \item If $\lnsh$ is not fully saturated, enter the focusing phase.
        \end{enumerate}
    \end{enumerate}
    \item[$\triangleright$] \textbf{Focusing Phase.} Apply exactly one instance of $\mathsf{fb}$ when $\lnsh$ is unfocused, and exactly one instance of $\mathsf{ub}$ when $\lnsh$ is focused. At every premise, enter the logical phase.
    \item[$\triangleright$] \textbf{Logical Phase.} Apply exactly one logical rule to a complex formula in the end component of $\lnsh$. Enter the saturation phase at every premise.
\end{itemize}
\end{definition}

A few comments about this definition. First, note that the focusing phase always terminates, even if $\lnsh$ does not contain a formula of the form $A \until B$. In this case $\lnsh$ is unfocused, the rule $\mathsf{fb}$ is applied and the instance has a single premise which is a copy of $\lnsh$. Second, if the construction enters the logical phase, it is guaranteed that the end component of $\lnsh$ contains a complex formula, since $\lnsh$ is not fully saturated. Third, the logical phase may apply an expansion rule or an implication rule to a formula in the end component. An expansion rule may produce a premise whose newly interior component contains complex formulae, since the previous end component is now an interior component. An implication rule leaves the interior components unchanged; consequently, every non-initial premise remains saturated.

\begin{lemma}[Local-wellformedness]\label{lem: well-formedness}
    Let $\lnsh$ be an LNS.
    \begin{enumerate}
        \item Every complex formula in $\lnsh$ admits a succinct bottom-up application of its corresponding logical rule, using the interior or end component variant as appropriate. Every premise of this application is an LNS.
        \item If $\lnsh$ is unfocused, then $\mathsf{fb}$ is applicable and all its premises are LNSs.
        \item If $\lnsh$ is focused, then $\mathsf{ub}$ is applicable and all its premises are LNSs.
    \end{enumerate}
\end{lemma}
\begin{proof}
    By inspection of the rules.
\end{proof}

We distinguish between two types of open leaves in the construction. An \emph{S-leaf} is an open leaf at which the saturation phase is started. An \emph{L-leaf} is an open leaf at which the logical phase is started.

\begin{proposition}\label{prop: existence proof-search tree}
    For every LNS $\lnsg$, the phased construction initiated with $\lnsg$ produces a pseudo-derivation of $\lnsg$, in which every leaf is either initial or fully saturated.
\end{proposition}
\begin{proof}
    Let $\pi_0$ be the tree consisting of a single node labelled $\lnsg$, with its root marked $S$. We construct an infinite sequence of finite trees
    \begin{equation*}
        \pi_0 \subseteq \pi_1 \subseteq \pi_2 \subseteq \ldots
    \end{equation*}
    by processing all open leaves at each stage. We prove by induction on $n$ that the construction of $\pi_n$ is (I) well-defined and satisfies the following invariant:
    \begin{enumerate}
        \item[(1)] Every node of $\pi_n$ is labelled by an LNS.
        \item[(2)] Every node which is not a leaf is the conclusion of an instance of a rule of $\pcalc$ with its children the corresponding premises.
        \item[(3)] Every open leaf marked L is non-initial, saturated, and not fully saturated.
    \end{enumerate}
    \textsc{Base case.} By definition $\pi_0$ consists of a single node labelled with $\lnsg$, and $\lnsg$ is an LNS by assumption. Moreover $\lnsg$ is an open leaf marked $S$, so (I) holds. \smallskip

    \noindent \textsc{Induction step.} Suppose $\pi_n$ satisfies (I). If $\pi_n$ does not contain open leaves, set $\pi_{n+1} := \pi_n$. Note that $\pi_{n+1}$ satisfies (I). Otherwise let $\lnsh$ be an open leaf of $\pi_n$. We distinguish the following cases. (i) If $\lnsh$ is an S-leaf and initial or fully saturated, the construction closes $\lnsh$, which introduces no new nodes and therefore preserves the invariant (I). (ii) If $\lnsh$ is an S-leaf, non-initial and not saturated, then some interior component contains a complex formula. By Lemma~\ref{lem: well-formedness}, the procedure can bottom-up apply the corresponding logical rule to said formula and the obtained premises are LNSs. Hence appending the premises and marking them S preserves (I). (iii) If $\lnsh$ is an S-leaf, non-initial, saturated, and not fully saturated, then either $\lnsh$ is unfocused or it contains a single focused formula. If $\lnsh$ is unfocused, then by Lemma~\ref{lem: well-formedness} the procedure can bottom-up apply $\mathsf{fb}$ and the resulting premises are LNSs. If $\lnsh$ contains a formula in focus, then by Lemma~\ref{lem: well-formedness} the procedure can bottom-up apply $\mathsf{ub}$ and the resulting premises are LNSs. Thus in either case appending the corresponding premises gives a well-formed tree. Moreover, every premise remains non-initial, saturated and not fully saturated, and so we mark each premise with L. Note that (I) is preserved. (iv) If $\lnsh$ is marked L, then $\lnsh$ is non-initial, saturated and not fully saturated by (I)(3). This implies that $\lnsh$ does not contain a complex formula in an interior component and thus must contain a complex formula in the end component. By Lemma~\ref{lem: well-formedness} the procedure can bottom-up apply the corresponding logical rule to the complex formula and the resulting premises are LNSs. Thus appending these premises and marking them S preserves invariant (I).

    Since each processed leaf has only finitely many premises, and $\pi_n$ has only finitely many open leaves, the resulting tree $\pi_{n+1}$ is finite. This finishes the induction.

    Now define 
    \begin{equation*}
        \pi := \bigcup_{n \in \mathbb{N}} \pi_n
    \end{equation*}
    and delete all phase markings. Note that the root of $\pi$ is labelled by $\lnsg$ and every node by an LNS. If a node is not a leaf, then it is a non-leaf node of some $\pi_n$ and therefore, by (I)(2) the node and its children form a correct rule instance of a rule of $\pcalc$. Since the construction never alters the children at any subsequent stage, the node and its children form a correct rule instance of a rule of $\pcalc$ in $\pi$. Therefore $\pi$ is a pseudo-derivation of $\lnsg$. Finally, every open leaf is processed at some stage of the construction, implying that every leaf of $\pi$ is closed, and therefore initial or fully saturated.
\end{proof}

\begin{definition}
    Let $\lnsg$ be an LNS. A pseudo-derivation of $\lnsg$ obtained by the phased construction is called a \emph{proof-search tree} for $\lnsg$.
\end{definition}

We will now show that every proof-search tree satisfies the SRP. The \emph{length} $\lgth{A}$ of a formula $A$ is defined to be the number of symbols it contains. We define $\lgth{A^{\an}} := \lgth{A}$ for $A \in \lang$ and $\an \in \set{\annou,\annof}$. For a set $\Gamma$ of annotated formulae we define $\lgth{\Gamma} := \sum_{A^\an \in \Gamma} \lgth{A^\an}$. Given a Gentzen sequent $\Gamma \sar \Delta$, we define the \emph{size} of $\Gamma \sar \Delta$ to be $s(\Gamma \sar \Delta) := \lgth{\Gamma} + \lgth{\Delta}$. Now let $\lnsg = \Gamma_0 \sar \Delta_0 \sep \ldots \sep \Gamma_{N-1} \sar \Delta_{N-1}$ and define the vector
\begin{equation*}
    s_N(\lnsg) := (s(\Gamma_0 \sar \Delta_0), \ldots, s(\Gamma_{N-1} \sar \Delta_{N-1})).
\end{equation*}
For the following lemma, $<_l$ denotes the standard lexicographical order on $\mathbb{N}^N$. It is well-known that $<_l$ is a well-order.

\begin{lemma}\label{lem: measure decrease} 
Let $\lnsh$ be a premise and $\lnsg$ the conclusion of a succinct logical rule instance of a rule $\mathsf{r}$ of $\pcalc$. If $\size{\lnsh} = \size{\lnsg} = N$, then $s_N(\lnsh) <_l s_N(\lnsg)$.
\end{lemma}
 \begin{proof}
   Suppose the principal formula occurs in component $i$. By inspection of the rules, note that all components $j < i$ remain unchanged in the rule instance of $\mathsf{r}$. In component $i$, since the rule instance is succinct, the principal formula is removed in the step from $\lnsg$ to $\lnsh$. If the principal formula is of the form $A \to B$ and $\mathsf{r}$ is an implication rule, then the auxiliary formulae introduced in component $i$ of $\lnsh$ are $A$ and/or $B$ and $\lgth{A} + \lgth{B} < \lgth{A \to B}$. If the principal formula is of the form $\X A$, and $\mathsf{r}$ is a next rule, then $\mathsf{r} \in \{\nextli, \nextri\}$ and the auxiliary formula $A$ is introduced in component $i+1$. If the principal formula is of the form $A \until B$ and $\mathsf{r}$ is an until rule, then the auxiliary formulae introduced in component $i$ are $A$ and/or $B$ and $\lgth{A} + \lgth{B} < \lgth{A \until B}$. Therefore $s(\lnsh(j)) = s(\lnsg(j))$ for all $j < i$ and $s(\lnsh(i)) < s(\lnsg(i))$, which implies $s_N(\lnsh) <_l s_N(\lnsg)$.
 \end{proof}

 \begin{corollary}\label{cor: measure decrease1}
     Every infinite branch $\branch$ of a proof-search tree $\pi$ passes through the logical phase infinitely often (and hence through the focusing phase infinitely often). 
 \end{corollary}
 \begin{proof}
     Suppose otherwise. Then the phased construction of $\pi$ passes only finitely many times through the logical phase, implying that a suffix of $\branch = \{\lnsg_i\}_{i \in \mathbb{N}}$ exists that contains only saturation steps. Note that the saturation phase only applies succinct length-preserving logical rules. Therefore there exists a suffix $\lnspath = \{\lnsg_i\}_{i \geq j}$ with $\lnsg_i$ the conclusion and $\lnsg_{i+1}$ a premise of a succinct logical rule and $\size{\lnsg_i} = \size{\lnsg_{i+1}} = N$ for all $i \geq j$. By Lemma~\ref{lem: measure decrease} the sequence $\{s_N(\lnsg_i)\}_{i \geq j}$ is an infinite and strictly decreasing sequence in $<_l$, contradicting the fact that $<_l$ is a well-order.  
 \end{proof}

 \begin{corollary}\label{cor: measure decrease2}
     For every infinite branch $\branch = \{\lnsg_i\}_{i \in \mathbb{N}}$ of a proof-search tree $\pi$, $\{\size{\lnsg_i}\}_{i \in \mathbb{N}}$ is unbounded.
 \end{corollary}
 \begin{proof}
    Suppose $\{\size{\lnsg_i}\}_{i \in \mathbb{N}}$ is bounded. Then there is a suffix $\lnspath = \{\lnsg_i\}_{i \geq j}$ of $\branch$ with $\size{\lnsg_i} = \size{\lnsg_{i+1}} = N$ for all $i \geq j$. If $\lnsg_i$ is the conclusion and $\lnsg_{i+1}$ the premise of a logical rule instance, then $s_N(\lnsg_{i+1}) <_l s_N(\lnsg_i)$ by Lemma~\ref{lem: measure decrease}. If $\lnsg_i$ is the conclusion and $\lnsg_{i+1}$ the premise of a branching focus rule, then $s_N(\lnsg_{i+1}) \leq_l s_N(\lnsg_i)$. Note that after every instance of a branching focus rule, there is at least one instance of a logical rule. Thus $\{s_N(\lnsg_i)\}_{i \geq j}$ is a decreasing infinite sequence in $<_l$ which strictly decreases infinitely often, contradicting that $<_l$ is a well-order.
 \end{proof}
 
\begin{lemma}\label{lem: SRP}
    Every proof-search tree satisfies the SRP.
\end{lemma}
\begin{proof}
Let $\branch = \{\lnsg_i\}_{i \in \mathbb{N}}$ be an infinite branch of $\pi$. By Corollary~\ref{cor: measure decrease1} $\branch$ passes infinitely often through the logical phase. In each logical phase the underlying LNS is saturated. Hence $\branch$ contains infinitely many saturated sequents. 
\end{proof}

Let \(A^\an\) be a complex formula occurrence in an LNS $\lnsg$ on a branch $\branch$ of a proof-search tree. We say that \emph{$A^\an$ is decomposed at $\lnsg$} if $\lnsg$ is the conclusion of a logical rule instance with $A^\an$ principal. Moreover, $A^\an$ is \emph{eventually decomposed on $\branch$} if there is a finite segment $\{\lnsg_i\}_{i \leq k}$ of $\branch$, with $\lnsg_0 = \lnsg$, along which the occurrence is carried through side occurrences and the annotation changes induced by the chosen premises of $\mathsf{fb}$ and $\mathsf{ub}$, until its resulting occurrence is decomposed at $\lnsg_k$. Before this decomposition, the occurrence remains in the same component and side of each LNS.

\begin{lemma}\label{lem: eventual decomposition}
    On an infinite branch $\branch$ of a proof-search tree, every complex formula is eventually decomposed on $\branch$. The same holds for every finite branch ending in a fully saturated sequent.
\end{lemma}
\begin{proof}
    We only consider the case for an infinite branch $\branch = \{\lnsg_i\}_{i \in \mathbb{N}}$. Let $i \in \mathbb{N}$ and suppose $\lnsg_i$ contains a complex formula $A^\an$ in component $j$. By Corollary~\ref{cor: measure decrease2} and Lemma~\ref{lem: SRP} there exists $k > i$ with $\lnsg_k$ saturated and $\size{\lnsg_k} > j+1$. The $j$-component of $\lnsg_k$ is therefore free of complex formulae, implying that $A^\an$ was decomposed at some node between $\lnsg_i$ and $\lnsg_k$.
\end{proof}

We now define two dual notions of subtree within a proof-search tree. Intuitively, a \emph{CB-structure} represents a potential proof of $\lnsg$, while a \emph{DB-structure} represents a potential counter-model.

\begin{definition}[CB- and DB-Closure]\label{def:CB-DB-closure}
Let $\prf$ be a proof-search tree for $\lnsg$. A \emph{CB-structure} is a subtree $(T, \prec)$ of $\prf$ such that $\lnsg \in T$ and for each non-leaf node $\lnsh \in \pi$ belonging to $T$:
\begin{itemize}

\item[$(1)$] If $\lnsh$ is the conclusion of a conjunctive rule in $\pi$, then all premises are in $T$;

\item[$(2)$] If $\lnsh$ is the conclusion of a disjunctive rule in $\pi$, then exactly one premise is in $T$.

\end{itemize}

\noindent
A \emph{DB-structure} is a subtree $(T, \prec)$ of $\prf$ such that $\lnsg \in T$ and for each non-leaf node $\lnsh \in \pi$ belonging to $T$:
\begin{itemize}

\item[$(1)$] If $\lnsh$ is the conclusion of a conjunctive rule in $\pi$, then exactly one premise is in $T$;

\item[$(2)$] If $\lnsh$ is the conclusion of a disjunctive rule in $\pi$, then all premises are in $T$.

\end{itemize}

\noindent
The \emph{CB-closure} $\cbclo(\prf)$ is the set of all CB-structures of $\prf$, and the \emph{DB-closure} $\dbclo(\prf)$ is the set of all DB-structures of $\prf$.
\end{definition}

We remark that since a CB-structure selects exactly one premise at each disjunctive node, and since both $\fbr$ and $\ubr$ include a premise identical to the conclusion, a CB-structure may contain instances of $\mathsf{triv}$. As noted above, such instances are admissible and can be eliminated by contracting the premise and conclusion into a single node. For the remaining (non-trivial) instances, a pruned $\fbr$ application reduces to an instance of $\annof$ (focusing a single formula), and a pruned $\ubr$ application reduces to an instance of $\annou$ (unfocusing a formula). For a CB-structure $\beta$, we write $N(\beta)$ for the tree obtained by contracting all unary instances of $\mathsf{triv}$, and we identify the remaining pruned $\fbr$ and $\ubr$ applications with $\annof$ and $\annou$ instances, respectively.

A branch in a proof-search tree is \emph{successful} if it either terminates at an initial sequent or is infinite and progressing; it is \emph{failed} if it either terminates at a fully saturated sequent or is infinite and not progressing. Note that every branch is either successful or failed (but not both). The proof of the following lemma is then immediate.

\begin{lemma}\label{lem: normalization}
Let $\pi$ be a proof-search tree for $\lnsg$, and let $\beta\in\cbclo(\pi)$. Then $N(\beta)$ is an $\calci$-derivation of $\lnsg$ satisfying the SRP. Moreover, $N(\beta)$ is a proof if and only if every branch of $\beta$ is successful.
\end{lemma}

\subsection{Completeness}

We now show how to use CB- and DB-structures to construct proofs and countermodels, respectively, which yields completeness. We first prove some auxiliary lemmas.

\begin{lemma}\label{lem:DB-CB-dual}
Let $\prf$ be a proof-search tree for $\lnsg$. Then, there exists $\beta \in \cbclo(\prf)$ such that $N(\beta)$ is a $\calci$-proof \iffi every DB-structure in $\dbclo(\prf)$ contains a successful branch.
\end{lemma}

\begin{proof}
We interpret the proof-search tree $\prf$ as the arena of a two-player game between \emph{Prover} and \emph{Refuter}:
\begin{itemize}

\item At conjunctive (binary) nodes, Refuter chooses a premise;

\item At disjunctive ($\fbr$ or $\ubr$) nodes, Prover chooses a premise;

\item At unary nodes, the move is forced.

\end{itemize}
A \emph{play} is a maximal path through $\prf$ determined by the players' choices. Prover wins a play if it is successful; Refuter wins otherwise. We observe the following correspondence between strategies and closures:
\begin{itemize}

\item A \emph{Prover strategy} determines a CB-structure: at disjunctive nodes, one premise is retained (Prover's choice), while at conjunctive nodes, both premises are retained. A Prover strategy is \emph{winning} \iffi every branch in the corresponding CB-structure is successful.

\item A \emph{Refuter strategy} determines a DB-structure: at conjunctive nodes, one premise is retained (Refuter's choice), while at disjunctive nodes, all premises are retained. A Refuter strategy is \emph{winning} \iffi the corresponding DB-structure has no successful branch.

\end{itemize}
Under this correspondence, we obtain:
\begin{align*}
\exists \beta \in\cbclo(\prf): N(\beta) \text{ is a proof} \;&\iff\; \exists \beta \in \cbclo(\prf):\text{ every branch of $\beta$ is successful} \\
\;&\iff\; \text{Prover has a winning strategy} \\
&\iff\; \text{Refuter has no winning strategy} \\
&\iff\; \text{every DB-structure has a successful branch}
\end{align*}
The first equivalence follows from Lemma~\ref{lem: normalization}. The second equivalence follows from the strategy--structure correspondence: a winning Prover strategy is precisely a CB-structure in which every branch is successful. The fourth equivalence is analogous: a winning Refuter strategy is precisely a DB-structure in which every branch is failed. It remains to justify the third equivalence, which is the determinacy of the game.

The winning condition for Prover is Borel: whether a finite play is successful depends only on the leaf reached (a condition determined by a finite prefix), while whether an infinite play is progressing is determined by countable conditions on the play's steps---specifically, the existence of a starting index from which (i)~every step preserves a formula in focus and (ii)~infinitely many steps pass through the right premise of $\untilli$ or $\untillii$ with that formula principal. Each of these is expressible as a countable intersection or union of conditions depending on finitely many steps of the play. By Martin's Borel Determinacy Theorem~\cite{Mar75}, the game is determined, which completes the proof.
\end{proof}

A \emph{trace value} is a pair $\tri = \bracket{A \until B^{\an}, i}$ where $A \until B \in \lang$, $\an \in \set{\annof, \annou}$, and $i \in \mathbb{N}$. Given a linear nested sequent $\lnsg = \Gamma_{0} \sar \Delta_{0} \sep \cdots \sep \Gamma_{n} \sar \Delta_{n}$, a trace value $\bracket{A \until B^{\an}, i}$ is a \emph{trace value of~$\lnsg$} if $A \until B^{\an} \in \Gamma_i$ for some $0 \leq i \leq n$. 

\begin{definition}
    Let $\ru$ be a rule of $\mathrm{LNS}^{+}_{\mathrm{LTL}}$,
    and fix an instance with conclusion $\lnsg$ and premise $\lnsh$.
    A pair $(\bracket{A \until B^{\an}, i},
         \bracket{A \until B^{\mathsf{b}}, j})$ is a \emph{trace pair} for $(\lnsg,\lnsh)$ \iffi
    $\bracket{A \until B^{\an}, i}$ is a trace value of $\lnsg$,
    $\bracket{A \until B^{\mathsf{b}}, j}$ is a trace value of
    $\lnsh$, and one of the following holds:
    \begin{enumerate}
        \item[$(1)$]
        $\ru$ is a logical rule, the occurrence of
        $A \until B^{\an}$ in the antecedent of component $i$
        is not principal in the rule instance,
        $\an = \mathsf{b}$, and $j = i$;

        \item[$(2)$]
        $\ru \in \set{\untilli,\untillii}$, the occurrence of
        $A \until B^{\an}$ in the antecedent of component $i$
        is principal in the rule instance,
        $\lnsh$ is its right premise,
        $\an = \mathsf{b}$, and $j = i+1$;

        \item[$(3)$]
        $\ru = \mathsf{fb}$,
        $\an = \annou$,
        $\mathsf{b} \in \set{\annou,\annof}$,
        and $j = i$;

        \item[$(4)$]
        $\ru = \mathsf{ub}$, $j = i$, and either
        $\lnsh = \lnsg$ and $\an = \mathsf{b}$,
        or $\lnsh = \lnsg^{\annou}$ and
        $\mathsf{b} = \annou$.
    \end{enumerate}
\end{definition}
A trace pair is \emph{progressing} if $j = i + 1$. A \emph{trace} along a path $\lnspath = \set{\lnsg_{i}}_{i \in \idxset}$ is a sequence of trace values $\set{\tau_{i}}_{i \in \idxset}$ such that $(\tau_i, \tau_{i+1})$ is a trace pair for $(\lnsg_{i}, \lnsg_{i{+}1})$ for each $i$. A trace along an infinite path is \emph{strongly progressing} if infinitely many of its trace pairs are progressing. An infinite path is \emph{strongly progressing} if it admits a strongly progressing trace, and an infinite branch is \emph{strongly progressing} if it contains a suffix that is a strongly progressing path.

\begin{lemma}\label{lem:no-prog-no-strongly-prog}
Let $\prf$ be a proof-search tree. If a DB-structure in $\dbclo(\prf)$ does not contain a progressing branch, then it contains a branch that is not strongly progressing.
\end{lemma}

\begin{proof}
    Let $\lnsg$ be the root of $\pi$, and let
    $\alpha \in \mathrm{DB}(\pi)$ contain no progressing branch. If $\alpha$ contains a finite branch, then that branch is not
    strongly progressing, and we are done. Hence assume that $\alpha$
    contains only infinite branches. In particular, every node of $\alpha$
    has at least one child. Let
    \begin{equation*}
        \mathcal{U}
        :=
        \{
            C \in \operatorname{sufo}(\lnsg)
            \mid C = A \until B \text{ for some } A,B \in \lang
        \}.
    \end{equation*}
    This set is finite. Moreover, inspection of the rules shows that
    every formula occurring in $\pi$, after erasing its annotation,
    belongs to $\operatorname{sufo}(\lnsg)$. Consequently, the
    underlying formula of every trace value occurring in $\pi$
    belongs to $\mathcal{U}$.

    If $\mathcal{U} = \varnothing$, then no node of $\pi$ has a trace
    value, so any branch of $\alpha$ is not strongly progressing.
    Thus assume that
    \[
        \mathcal{U} = \{C_0,\ldots,C_{M-1}\},
        \qquad M \geq 1,
    \]
    where the $C_j$ are pairwise distinct.

    We recursively construct an infinite branch
    \[
        \mathcal{B} = (\lnsg_n)_{n \in \mathbb{N}}
    \]
    of $\alpha$, starting with $\lnsg_0 := \lnsg$.
    Simultaneously, we maintain a counter $q_n \in \mathbb{N}$
    recording the number of $\mathsf{fb}$ applications encountered in the branch
    before reaching $\lnsg_n$. Initially, let $q_0 := 0$.

    Suppose that $\lnsg_n$ and $q_n$ have been defined. We distinguish
    three cases according to the rule applied at $\lnsg_n$.
    \begin{enumerate}
        \item[$(1)$]
        If the rule is conjunctive, let $\lnsg_{n+1}$ be its unique
        premise retained in $\alpha$, and set $q_{n+1} := q_n$.

        \item[$(2)$]
        If the rule is $\mathsf{ub}$, let $\lnsg_{n+1}$ be the
        unchanged premise, so that the current focus is retained,
        and set $q_{n+1} := q_n$.

        \item[$(3)$]
        If the rule is $\mathsf{fb}$, let $j := q_n \bmod M$. If $C_j^{\annou}$ occurs in an antecedent of $\lnsg_n$,
        choose the premise that focuses this occurrence. Otherwise,
        choose the unchanged premise. In either case, set $q_{n+1} := q_n + 1$.
    \end{enumerate}

    The choice in the third case is unambiguous. Indeed, the
    conclusion of an $\mathsf{fb}$ instance is unfocused and, by
    the phased construction, saturated. Its interior components
    therefore contain no until-formulae. Hence $C_j^{\annou}$,
    if present in an antecedent, occurs only in the end component.

    All these choices are available in $\alpha$: it retains exactly
    one premise at every nonterminal conjunctive node and all
    premises at every disjunctive node. Since $\alpha$ contains no
    finite branch, the construction continues indefinitely and
    defines an infinite branch $\mathcal{B}$. Every node of this
    branch has the same rule application as in $\pi$, so
    $\mathcal{B}$ is also an infinite branch of $\pi$.

    We establish three claims. \smallskip
    
    \noindent \emph{Claim 1. $\forall n \in \mathbb{N}\;
        \exists k \geq n\;
        \lnsg_k \text{ is unfocused}$.} \smallskip
        
 Suppose otherwise. Then there exists $n \in \mathbb{N}$ such that every $\lnsg_k$, for $k \geq n$, contains a focused
    formula. Let $\tau_k$ be the unique focused trace value
    of $\lnsg_k$. For every $k \geq n$, the pair $(\tau_k,\tau_{k+1})$ is a trace pair: $\mathsf{fb}$ cannot be applied to a focused sequent, and at
    every $\mathsf{ub}$ instance our construction selects the
    unchanged premise. Moreover, in a logical rule instance, the focused
    occurrence either remains a side occurrence or is principal
    in an instance of $\untilli$ or $\untillii$. In the latter
    case, the branch must take the right premise, since the
    left premise is unfocused by Lemma~\ref{lem: property until rule}.

    By Lemma~\ref{lem: eventual decomposition}, the focused occurrence at any $\lnsg_k$,
    with $k \geq n$, is eventually decomposed. Until that
    decomposition, it remains focused. Thus, for every
    $k \geq n$, there exists $\ell \geq k$ such that the
    focused occurrence is principal in an instance of
    $\untilli$ or $\untillii$ at $\lnsg_\ell$. As observed
    above, the branch takes the right premise of that instance. Consequently, the suffix $(\lnsg_k)_{k \geq n}$ has a
    continuous focus trace and passes infinitely often through
    the right premise of an until-left rule with the focused
    formula principal. Therefore $\mathcal{B}$ is progressing,
    contrary to the assumption on $\alpha$. This proves Claim~1. \smallskip
    
    \noindent \emph{Claim 2. The branch $\mathcal{B}$ passes through
    infinitely many $\mathsf{fb}$ instances.}\smallskip

    Fix an arbitrary $n \in \mathbb{N}$. By
    Corollary~\ref{cor: measure decrease1}, there is a focusing-phase application
    at some $\lnsg_s$ with $s \geq n$. If that application is $\mathsf{fb}$, we have found
    an $\mathsf{fb}$ instance at an index at least $n$. Otherwise, the application is $\mathsf{ub}$, so $\lnsg_s$ is focused. By Claim~1, there exists $k>s$ such that $\lnsg_k$ is unfocused. Let $\lnsg_t$,
    with $t \geq k$, be the first subsequent conclusion of a focusing-phase application. Such a node exists
    by Corollary~\ref{cor: measure decrease1}.

    Between $\lnsg_k$ and $\lnsg_t$, only logical rules are applied. A logical rule cannot introduce a focused formula into an unfocused LNS. Hence $\lnsg_t$ is unfocused, and the focusing rule applied there must be $\mathsf{fb}$. We have therefore found an $\mathsf{fb}$ instance
    at an index at least $n$. Since $n$ was arbitrary, there are infinitely many such instances. \smallskip
    
    \noindent \emph{Claim 3. No suffix of $\branch$ admits an infinite trace.} \smallskip

    Suppose, toward a contradiction, that $\tau = (\tau_n)_{n \geq m}$ is an infinite trace along a suffix of
    $\branch$. By the definition of trace pairs,
    the underlying until-formula is constant along
    $\tau$. Thus, for some $j<M$,
    \begin{equation*}
        \tau_n = \bracket{C_j^{a_n},i_n}
        \qquad\text{for every } n \geq m.
    \end{equation*}
    By Claim~2, we can enumerate the indices of the $\mathsf{fb}$ applications along $\mathcal{B}$ as
    \begin{equation*}
        s_0 < s_1 < s_2 < \cdots.
    \end{equation*}
    By construction, $q_{s_r}=r$ for every
    $r \in \mathbb{N}$. Choose $r$ such that
    \begin{equation*}
        s_r \geq m
        \qquad\text{and}\qquad
        r \bmod M = j.
    \end{equation*}
    Such an $r$ exists because there are infinitely many natural numbers congruent to $j$ modulo $M$. The sequent $\lnsg_{s_r}$ is unfocused and
    saturated. Since $\tau_{s_r}$ is a trace value of this sequent, we have
    \begin{equation*}
        a_{s_r} = \annou,
        \qquad
        i_{s_r} = \|\lnsg_{s_r}\|-1.
     \end{equation*}
    In particular, $C_j^{\annou}$ occurs in the
    antecedent of the end component of $\lnsg_{s_r}$.

    At this $\mathsf{fb}$ application, our construction of $\branch$ selects $C_j$. Therefore
    $\lnsg_{s_r+1}$ is the premise obtained by
    replacing $C_j^\annou$ with $C_j^\annof$. Hence, by the definition
    of trace pairs,
    \begin{equation*}
        \tau_{s_r+1}
        =
        \bracket{C_j^{\annof},\|\lnsg_{s_{r}+1}\|-1}.
     \end{equation*}

    We show by induction on $n \geq s_r+1$ that
    $\tau_n$ is focused. The base case was just
    established. For the inductive step, suppose that $\tau_n$ is focused. Then $\lnsg_n$ is focused, so the rule applied there cannot be $\mathsf{fb}$. If it is $\mathsf{ub}$, our construction selects the unchanged premise, and the trace-pair definition preserves the annotation. If it is a logical rule, every
    trace pair preserves the annotation as well. Thus $\tau_{n+1}$ is focused. Consequently, every node of the suffix $(\lnsg_n)_{n \geq s_r+1}$ is focused, contradicting Claim~1. This proves
    Claim~3.

    In particular, no suffix of $\mathcal{B}$
    admits a strongly progressing trace. Therefore
    $\mathcal{B}$ is not strongly progressing,
    as required.
\end{proof}

We now show that every LNS either has an $\calci$-proof satisfying the SRP or is falsifiable, from which completeness follows immediately.

\begin{lemma}\label{lem: proof or falsifiable} Let $\lnsg$ be an LNS. Then, either $\lnsg$ has an $\calci$-proof satisfying the SRP or there exists a state sequence $\Bar{\sigma}$ such that $\Bar{\sigma} \not \sat \lnsg$.
\end{lemma}

\begin{proof} Let $\lnsg$ be an LNS and let $\pi$ be a proof-search tree for $\lnsg$. If there exists $\alpha \in \cbclo(\prf)$ such that $N(\alpha)$ is a $\calci$-proof of $\lnsg$, then it satisfies the SRP by Lemma~\ref{lem: SRP}. Suppose otherwise that for every $\alpha \in \cbclo(\prf)$, $N(\alpha)$ is not a $\calci$-proof. By Lemma~\ref{lem:DB-CB-dual}, some DB-structure $\alpha$ in $\dbclo(\prf)$ contains only failed branches. In particular, $\alpha$ contains no progressing branch, so by Lemma~\ref{lem:no-prog-no-strongly-prog}, $\alpha$ contains a branch $\branch$ that is not strongly progressing. Either $\branch$ is finite, in which case failure of the branch implies it must end in a fully saturated leaf, or infinite and not strongly progressing. We treat only the latter case, as the former is argued similarly. Let $\branch := \set{\lnsg_{i}}_{i \in \mathbb{N}}$ be an infinite branch of $\alpha$ that is \emph{not} strongly progressing.

We use $\branch$ to construct a state sequence $\sts$ such that $\sts \not\sat \lnsg$. Define:
$$
\Gamma_{i} := \!\!\!\!\!\!\!\!\!\!\!\!\!\!\!\! \bigcup_{j \in \mathbb{N}, \, \lnsg_{j}(i) \, = \, (\Gamma \, \sar \, \Delta)} \!\!\!\!\!\!\!\!\!\!\!\!\!\!\!\! \Gamma^{-}
\qquad
\Delta_{i} := \!\!\!\!\!\!\!\!\!\!\!\!\!\!\!\! \bigcup_{j \in \mathbb{N}, \, \lnsg_{j}(i) \, = \, (\Gamma \, \sar \, \Delta)} \!\!\!\!\!\!\!\!\!\!\!\!\!\!\!\! \Delta^{-}
$$
That is, $\Gamma_i$ collects all unannotated formulae appearing in the antecedent of the $i$-component along $\branch$, and $\Delta_i$ collects those in the consequent. We define $\sigma_{i} := \set{p \mid p \in \Gamma_{i}}$ and $\bar\sigma := \tuple{\sigma_0,\sigma_1,\ldots}$. We prove the following by mutual induction on the length of $A$:
\begin{enumerate}
\item[(1)] If $A \in \Gamma_{i}$, then $\sts, i \sat A$;
\item[(2)] If $A \in \Delta_{i}$, then $\sts, i \not\sat A$.
\end{enumerate}

\case{Case $A = \bot \in \Gamma_i$.} There exists $m \in \mathbb{N}$ such that $\bot \in \lnsg_m(i)$, making $\lnsg_m$ an initial sequent. Therefore, once the phased construction enters the saturation phase for the first time after $\lnsg_m$, the open leaf is closed as an initial node, contradicting the assumption that $\branch$ is infinite. Hence $\bot \not \in \Gamma_i$.

\case{Case $A = \bot \in \Delta_i$.} By definition, $\sts, i \not \models \bot$.

\case{Case $A = p \in \Gamma_{i}$.} By definition, $p \in \sigma_{i}$, so $\sts, i \sat p$.

\case{Case $A = p \in \Delta_{i}$.} Suppose for contradiction that $p \in \Gamma_{i}$ as well. Then there exist $m_0, m_1 \in \mathbb{N}$ with $\lnsg_{m_0}(i) = \Sigma \sar \Pi$ and $\lnsg_{m_1}(i) = \Sigma' \sar \Pi'$ such that $p \in \Sigma$ and $p \in \Pi'$. Without loss of generality, assume $m_0 \leq m_1$. Since no rule in $\calci$ deletes atoms, $p \in \Sigma$ implies $p \in \Sigma'$. Therefore $p \in \Sigma' \cap \Pi'$, making $\lnsg_{m_1}$ an instance of $\id$. By the same argument as in case $\bot \in \Gamma_i$, this implies that the branch $\branch$ is finite, contradicting our assumption. Hence $p \not\in \Gamma_{i}$, and so $\sts, i \not \sat p$.

\case{Case $A = B \to C \in \Delta_i$.} There exists $m \in \mathbb{N}$ such that $B \to C \in \Pi$ with $\lnsg_m(i) = \Sigma \sar \Pi$. Lemma~\ref{lem: eventual decomposition} implies that $B \to C$ is eventually decomposed. Hence, for some $l > m$ with $\lnsg_l(i) = \Sigma' \sar \Pi'$, we have $B \in \Sigma'$ and $C \in \Pi'$, giving $B \in \Gamma_i$ and $C \in \Delta_i$. By IH, $\Bar{\sigma}, i \sat B$ and $\Bar{\sigma}, i \not \sat C$, so $\Bar{\sigma}, i \not \sat B \to C$. The case for $A = B \to C \in \Gamma_i$ is similar and omitted.

\case{Case $A = \X B \in \Gamma_{i}$.} There exists $m \in \mathbb{N}$ such that $\X B \in \Sigma$ with $\lnsg_{m}(i) = \Sigma \sar \Pi$. By Lemma~\ref{lem: eventual decomposition}, $\X B$ is eventually decomposed. Hence, for some $l > m$ with $\lnsg_l(i+1) = \Sigma' \sar \Pi'$, we have $B \in \Sigma'$, giving $B \in \Gamma_{i+1}$. By IH, $\sts, i+1 \sat B$, and so $\sts, i \sat \X B$. The case for $A = \X B \in \Delta_i$ is similar and omitted.

\case{Case $A = B \until C \in \Gamma_{i}$.}
There exist $m \in \mathbb{N}$ and an annotation $\an_m \in \{\annou,\annof\}$ such that $(B \until C)^{a_m}$ occurs in the antecedent of component $i$
of $\lnsg_m$. We follow this particular occurrence along $\branch$,
allowing its annotation to change at branching focusing steps. More precisely, starting with
\[
    \tau_m := \langle (B \until C)^{\an_m},i\rangle,
\]
we recursively construct trace values
\[
    \tau_n = \langle (B \until C)^{\an_n},j_n\rangle
    \qquad (n \geq m)
\]
until the tracked occurrence is principal in an until-left inference whose left premise is selected by $\branch$. At each stage for which $\tau_n$ is defined, we maintain
the following invariant:
\begin{itemize}
    \item[(i)] $(\tau_k)_{m \leq k \leq n}
        \text{ is a trace along }
        (\lnsg_k)_{m \leq k \leq n}$;
    \item[(ii)] $j_n \geq i$ and $B \in \Gamma_\ell \text{ for every } i \leq \ell < j_n$.
\end{itemize}
The invariant holds initially: $\tau_m$ is a trace value of $\lnsg_m$, $j_m=i$, and the last condition is vacuous.

Suppose $\tau_n$ has been defined. We distinguish the following cases according to the rule applied at $\lnsg_n$. If the rule is logical and the occurrence represented by
$\tau_n$ is not principal, then it is carried as a side occurrence into $\lnsg_{n+1}$. Set $\tau_{n+1}:=\tau_n$. This gives a trace pair by the side-occurrence clause of
the trace-pair definition. The component index is unchanged, so the invariant is preserved.

If the rule is $\mathsf{fb}$ or $\mathsf{ub}$, follow the annotation change, if any, induced by the premise selected by $\branch$. Let $\an_{n+1}$ be the resulting annotation
and set $\tau_{n+1} := \langle (B \until C)^{\an_{n+1}},j_n\rangle$. If unfocusing merges the tracked occurrence with an
already-present unfocused entry, continue with the resulting entry. In every case, $(\tau_n,\tau_{n+1})$ is a trace pair
by the corresponding branching-focusing clause. Since the component index remains unchanged, the invariant is preserved.

Finally, suppose the tracked occurrence is principal in a logical rule instance. Since the formula is of the form $B \until C$ and occurs in the antecedent,
the rule is $\untilli$ or $\untillii$. If $\branch$ takes the left premise, then $C$ occurs in the antecedent of component $j_n$ of $\lnsg_{n+1}$. Consequently, $C \in \Gamma_{j_n}$, and we stop the construction.

If $\branch$ takes the right premise, then $B$ occurs in the antecedent of component $j_n$ of $\lnsg_{n+1}$, and $(B \until C)^{\an_n}$ occurs in the antecedent of
component $j_n+1$. Hence $B \in \Gamma_{j_n}$. Set $\tau_{n+1} := \langle (B \until C)^{\an_n},j_n+1\rangle$. The pair $(\tau_n,\tau_{n+1})$ is progressing. Moreover,
the previous invariant, together with
$B \in \Gamma_{j_n}$, gives
\[
    B \in \Gamma_\ell
    \qquad\text{for every } i \leq \ell < j_{n+1}.
\]
Thus the invariant is again preserved.

We claim that this construction stops after finitely many
steps. Suppose otherwise. It then defines an infinite trace
\[
    \tau=(\tau_n)_{n \geq m}
\]
along a suffix of $\branch$. This trace must have infinitely many progressing pairs.
Indeed, suppose it had only finitely many. Choose $N \geq m$ such that no pair $(\tau_n,\tau_{n+1})$, for $n \geq N$,
is progressing. The tracked occurrence cannot be principal in any logical inference at or after $\lnsg_N$: a right-premise
until-left step would be progressing, while a left-premise until-left step would stop the construction. Consequently, the occurrence represented by $\tau_N$ would
never be logically decomposed, contrary to Lemma~\ref{lem: eventual decomposition}. 

Thus $\tau$ is strongly progressing. This contradicts the choice of $\branch$ as a branch that is not strongly
progressing. The construction therefore stops at some index $s \geq m$, where the tracked occurrence is principal
in an until-left rule instance and $\branch$ takes its left premise.

Let $j:=j_s$. By the invariant and the stopping condition,
\[
    B \in \Gamma_\ell
    \quad\text{for every } i \leq \ell < j,
    \qquad
    C \in \Gamma_j.
\]
Hence the induction hypotheses yield
\[
    \bar{\sigma},\ell \models B
    \quad\text{for every } i \leq \ell < j,
    \qquad
    \bar{\sigma},j \models C.
\]
Since $j \geq i$, the state index $j$ witnesses $\bar{\sigma},i \models B \until C$.

\case{Case $A = B \until C \in \Delta_{i}$.} 
There exists
$m \in \mathbb{N}$ such that $B \until C$ occurs in the
consequent of component $i$ of $\lnsg_m$. Again, we follow this particular occurrence and its successive residuals. Since trace values are defined only for antecedent until-occurrences, we describe the consequent-occurrence construction directly, without
calling it a trace.

Initially, set $j_m:=i$ and select the given occurrence in $\lnsg_m$. At each stage $n$, the selected occurrence is an occurrence of $B \until C$ in the consequent of component $j_n$ of $\lnsg_n$. We maintain the following invariant:
\begin{itemize}
    \item[(i)] $j_n \geq i$;
    \item[(ii)] $ C \in \Delta_\ell
    \text{ for every } i \leq \ell < j_n.$
\end{itemize}
The invariant holds initially, since $j_m=i$. Suppose the selected occurrence has been defined at $\lnsg_n$. If it is not principal in the rule applied there, follow its unchanged occurrence in the consequent
of the same component of $\lnsg_{n+1}$, and set $j_{n+1}:=j_n$. This includes branching focusing steps, since
$\mathsf{fb}$ and $\mathsf{ub}$ leave all consequents unchanged. The invariant is preserved.

If the selected occurrence is principal, then the rule is $\untilri$ or $\untilrii$. If $\branch$ takes the left premise, then both $B$
and $C$ occur in the consequent of component $j_n$ of $\lnsg_{n+1}$. Thus $B,C \in \Delta_{j_n}$, and we stop the construction. If $\branch$ takes the right premise, then $C$
occurs in the consequent of component $j_n$ of
$\lnsg_{n+1}$, and $B \until C$ occurs in the consequent of component $j_n+1$. Hence $C \in \Delta_{j_n}$. Select this residual occurrence of $B \until C$ and set $j_{n+1}:=j_n+1$. If the residual coincides with an already-present
consequent entry, continue with that entry. The invariant is preserved, since it now also holds at the additional index $j_n$. There are two possibilities.

\smallskip
\noindent
\emph{Subcase 1: the construction continues indefinitely.}
In this case the selected occurrence is principal infinitely often. Otherwise, after its last decomposition, a selected residual would never be decomposed, contrary
to Lemma~\ref{lem: eventual decomposition}. Every such decomposition is an until-right inference whose right premise is selected, since selecting its left premise would stop the construction. Therefore
the sequence $(j_n)_{n \geq m}$ increases by one infinitely often and is otherwise constant. Since $j_m=i$, it is unbounded. Fix any $\ell \geq i$. Choose $n \geq m$ such that
$j_n>\ell$. By the invariant, $C \in \Delta_\ell$. Since $\ell$ was arbitrary, this holds for every $\ell \geq i$. Applying the induction hypothesis
to the proper subformula $C$, we obtain
\[
    \bar{\sigma},\ell \not\models C
    \qquad\text{for every } \ell \geq i.
\]
Thus there is no future state witnessing $C$, and consequently $\bar{\sigma},i \not\models B \until C$.

\smallskip
\noindent
\emph{Subcase 2: the construction stops after finitely
many steps.}
Let $s \geq m$ be the stopping index and put $j:=j_s$. The branch takes the left premise of the until-right inference applied to the selected occurrence at $\lnsg_s$. By the invariant and the stopping condition,
\[
    C \in \Delta_\ell
    \quad\text{for every } i \leq \ell \leq j,
    \qquad
    B \in \Delta_j.
\]
The induction hypotheses therefore give
\[
    \bar{\sigma},\ell \not\models C
    \quad\text{for every } i \leq \ell \leq j,
    \qquad
    \bar{\sigma},j \not\models B.
\]
Suppose, toward a contradiction, that
$\bar{\sigma},i \models B \until C$.
Then there exists $h \geq i$ such that
\[
    \bar{\sigma},h \models C
    \qquad\text{and}\qquad
    \bar{\sigma},\ell \models B
    \text{ for every } i \leq \ell < h.
\]
If $h \leq j$, the first condition contradicts
$\bar{\sigma},h \not\models C$.
If $h>j$, then $i \leq j<h$, so the second condition implies $\bar{\sigma},j \models B$, contradicting $\bar{\sigma},j \not\models B$.
Both possibilities lead to a contradiction. Hence $\bar{\sigma},i \not\models B \until C.$
\medskip

\noindent Finally, let $\lnsg(i) = \Sigma_{i} \sar \Pi_{i}$ for $0 \leq i < \size{\lnsg}$. Since $\Sigma_{i}^{-} \subseteq \Gamma_{i}$ and $\Pi_{i}^{-} \subseteq \Delta_{i}$ for each $0 \leq i < \size{\lnsg}$, it follows that $\sts, i \not\sat \lnsg(i)$ for all $0 \leq i < \size{\lnsg}$. Hence, by Definition~\ref{def: satisfaction LNS}, $\Bar{\sigma} \not \sat \lnsg$.
\end{proof}

\begin{theorem}[Completeness of $\calci$]\label{thm:completeness-non-wf}
If $\lnsg$ is valid, then $\lnsg$ has an $\calci$-proof satisfying the SRP.
\end{theorem}
\begin{proof}
    Suppose $\lnsg$ is valid. By Lemma~\ref{lem: proof or falsifiable} either (i) $\lnsg$ has a $\calci$-proof which satisfies the SRP or (ii) there exists a state sequence $\Bar{\sigma}$ such that $\Bar{\sigma}\not \models \lnsg$. Case (ii) cannot happen since $\lnsg$ is valid. Hence $\lnsg$ has a $\calci$-proof which satisfies the SRP.
\end{proof}

As a consequence of soundness, completeness and semantic invertibility (Lemma~\ref{lem: semantic invertibility}), we obtain the invertibility of all rules of $\calci$.
\begin{lemma}[Invertibility]\label{lem: invertibility}
    All rules $\ru$ of $\calci$ are invertible: for every instance of $\ru$ with conclusion $\lnsg$ and premises $\lnsh_i$ for $1 \leq i \leq n$, if $\lnsg$ has a $\calci$-proof, then $\lnsh_i$ has an $\calci$-proof for every $1 \leq i \leq n$.
\end{lemma}



\section{Correspondence between Non-Wellfounded and Cyclic Proofs}\label{sec:correspondence}

We now establish bi-directional proof transformations between non-wellfounded and cyclic proofs. We first address cycle recognition in non-wellfounded proofs. In order to turn non-wellfounded proofs into cyclic proofs, we will make use of the SRP, which guarantees that infinite branches contain suitable `repetitions' where they can be pruned and replaced by cycles.

Transforming cyclic proofs into non-wellfounded ones is non-trivial in the LNS setting. The key idea is captured by the Shifting Lemma (Lemma~\ref{lem:shifting}). Given a cyclic proof $(\der,L,c)$ whose conclusion $c(\lnsg)$ is the companion of a bud $\lnsg$, we construct a new cyclic proof of $\lnsg$ by shifting inferences forward in $(\der,L,c)$, yielding a cyclic proof $(\der',L',c')$. Iterating this construction allows one to unravel cycles ad infinitum, thereby producing a non-wellfounded derivation from a cyclic one. By carefully handling annotations in the construction, we can then show that the derivation is a non-wellfounded proof which satisfies the SRP. Together, these bi-directional transformations establish the soundness and completeness of $\calccyc$.

\subsection{From Non-Wellfounded Proofs to Cyclic Proofs}

This section shows how non-wellfounded proofs which satisfy the SRP can be transformed into cyclic proofs. We call a rule $\mathsf{r} \in \set{\nextlii, \nextrii, \untillii, \untilrii}$ an \emph{expansion rule} (cf.~\cite{LyoOst26}). We first show that every infinite branch in a non-wellfounded proof passes infinitely often through an expansion rule. We then show that every infinite branch in a non-wellfounded proof which satisfies the SRP contains a suitable `repetition' and finally use these two results to establish the translation.

\begin{lemma}\label{lem: nwf branches grow}
If $\pi$ is an $\calci$-proof and $\branch$ is an infinite branch of $\pi$, then $\branch$ passes infinitely often through $\nextlii$, $\nextrii$, or through the right premise of $\untillii$ or $\untilrii$.
\end{lemma}
\begin{proof}
 Let $\pi$ be an $\calci$-proof and suppose, toward a contradiction, that $\branch = \set{\lnsg_i}_{i \in \mathbb{N}}$ is an infinite branch of $\pi$, which does not pass infinitely often through an expansion rule. Then, there exists a $j \in \mathbb{N}$ such that the suffix $\lnspath=\set{\lnsg_i}_{i \geq j}$ of $\branch$ starting at $\lnsg_j$ never passes through instances of $\nextlii$, $\nextrii$ or the right premises of $\untillii$ or $\untilrii$. Therefore, $\size{\lnsg_i} = \size{\lnsg_k}$ for all $i,k \geq j$. Since on a continuously focused path, the component index of the focused formula is unchanged by non-progressing steps and increases by one at each progressing step, a progressing suffix cannot have bounded LNS length. Therefore $\branch$ does not have a progressing suffix and is thus not progressing. This contradicts our assumption that $\pi$ is a proof.
\end{proof}

\begin{lemma}\label{lem: suitable repetitions exist} Let $\prf$ be a non-wellfounded proof which satisfies the SRP and let $\branch$ be an infinite branch. Then, $\branch$ contains two saturated LNSs $\lnsg$ and $\lnsh$ such that $\lnsg$ occurs strictly below $\lnsh$ in $\branch$, $\size{\lnsg} < \size{\lnsh}$, the path from $\lnsg$ to $\lnsh$ contains a continuous focus trace and goes through the right premise of $\mathsf{r} \in \set{\untilli, \untillii}$ with the formula in focus principal, and $\last{\lnsg} = \last{\lnsh}$.
\end{lemma}

\begin{proof} Let $\pi$ be an $\calci$-proof of $\lnsf$ satisfying the SRP. All rules of $\calci$ are \emph{analytic}, i.e., if $\lnsk$ is the conclusion and $\lnsk'$ a premise of a rule, then $\sufo{\lnsk'} \subseteq \sufo{\lnsk}$.  Thus, there are at most $2^{3 \times |\sufo{\lnsf}|}$ many Gentzen sequents that can occur as components of an LNS in $\prf$. Since $\pi$ satisfies the SRP, every infinite branch encounters infinitely many saturated sequents. Since there are at most $2^{3 \times |\sufo{\lnsf}|}$ many Gentzen sequents that can serve as an end component, every infinite branch must contain infinitely many saturated sequents with identical end components. Moreover, since $\pi$ is a proof, every infinite branch $\branch$ is progressing and thus contains a suffix with a continuous focus trace that passes infinitely often through the right premise of $\untilli$ or $\untillii$ with the formula in focus principal. Finally, by Lemma~\ref{lem: nwf branches grow}, for every infinite branch $\branch = \{\lnsg_i\}_{i \in \mathbb{N}}$ there are infinitely many $i \in \mathbb{N}$ such that $\size{\lnsg_i} < \size{\lnsg_{i+1}}$. Thus, every infinite branch must contain a pair of saturated sequents $\lnsg$ and $\lnsh$ such that $\size{\lnsg} < \size{\lnsh}$, the path from $\lnsg$ to $\lnsh$ contains a continuous focus trace and goes through the right premise of $\mathsf{r} \in \set{\untilli, \untillii}$ with the formula in focus principal, and $\last{\lnsg} = \last{\lnsh}$. 
\end{proof}

Lemma~\ref{lem: suitable repetitions exist} implies that every infinite branch of a proof $\pi$ which satisfies the SRP contains a suitable candidate for a cycle, which allows us to extract a cyclic proof. 

\begin{theorem}\label{thm:non-wf-to-cyclic}
Every non-wellfounded proof satisfying the SRP can be transformed into a cyclic proof.
\end{theorem}

\begin{proof} Let $\prf$ be an $\calci$-proof of $\lnsg$ which satisfies the SRP. Let $\prf_0$ be the subtree of $\prf$ obtained by pruning every infinite branch of $\prf$ at the lowermost sequent $\lnsh$ such that $\lnsh$ is saturated, there exists a saturated sequent $\lnsk$ strictly below $\lnsh$ with $\size{\lnsk} < \size{\lnsh}$, the path from $\lnsk$ to $\lnsh$ contains a continuous focus trace and passes through the right premise of $\mathsf{r} \in \set{\untilli, \untillii}$ with the formula in focus principal, and $\last{\lnsk}=\last{\lnsh}$. For each such $\lnsh$ we call $\lnsk$ its companion. Since $\prf$ is finitely branching, K\H{o}nig's Lemma and Lemma~\ref{lem: suitable repetitions exist} imply that $\prf_0$ is finite and therefore a finite derivation of $\lnsg$. Let $L$ be the set of leaves of $\prf_0$ that were obtained by pruning a branch and let $c$ be the function which maps each such leaf to its companion. Observe that every other leaf which is not contained in $L$ is also a leaf of $\prf$ and therefore an initial sequent. Hence, the tuple $(\pi_{0}, L, c)$ is a cyclic proof of $\lnsg$.
\end{proof}

We obtain the completeness of the cyclic LNS system as a corollary of Theorems~\ref{thm:completeness-non-wf} and~\ref{thm:non-wf-to-cyclic}.

\begin{corollary}[Completeness of $\calccyc$]\label{cor:cyclic-complete}
 If $\lnsg$ is valid, then $\lnsg$ has an $\calccyc$-proof.
\end{corollary}

\subsection{From Cyclic Proofs to Non-Wellfounded Proofs}\label{subsec:unravel}

In this section, we prove the soundness of cyclic proofs by means of a proof transformation. We introduce a technique for unraveling cyclic linear nested sequent proofs; in particular, we show how a cyclic proof rooted at the companion of a cycle can be transformed into a cyclic proof of the corresponding bud (Lemma~\ref{lem:shifting}). This is achieved by shifting rule applications forward from the source proof to the target proof while maintaining a suitable invariant, ensuring that the transformation can be iterated indefinitely. By iterating this shifting procedure, we obtain an infinite sequence of growing cyclic proofs, whose limit is a non-wellfounded derivation. We then show that the limit is in fact a non-wellfounded proof which satisfies the SRP. The main difficulty for the last step is to prove that every infinite branch of the limit derivation is progressing. To show this, we will identify every infinite branch in the limit with an infinite \emph{looping path} of the original cyclic proof. Intuitively, a looping path is a path through a cyclic proof that, upon reaching a bud, may continue at its companion. Due to the use of annotations, we can guarantee that every infinite looping path in a cyclic proof is progressing. This property transfers to the corresponding infinite branch in the limit derivation, thus showing that the limit is a non-wellfounded proof. We begin by introducing looping paths and by showing that every infinite looping path has a suffix with a continuous focus trace. 

\begin{definition}\label{dfn:looping-path}
    Let $(\pi, L, c)$ be a cyclic proof. A \emph{looping path} of $(\pi, L, c)$ is a finite or infinite sequence $\lnspath^{\circ} = \set{\lnsg_i}_{i \in I}$ (where $I$ is an index set) of LNSs in $\pi$ such that the following hold:
    \begin{enumerate}
        \item[$(1)$] $\lnsg_0$ is the root of $\pi$;
        \item[$(2)$] If $\lnsg_i$ is a leaf of $\pi$ and $\lnsg_i \not \in L$, then $\lnspath^\circ$ ends at $\lnsg_i$;
        \item[$(3)$] If $\lnsg_i \in L$, then either $\lnspath^\circ$ ends at $\lnsg_i$ or $\lnsg_{i+1}$ is a child of $c(\lnsg_i)$;
        \item[$(4)$] Otherwise, $\lnsg_{i+1}$ is a child of $\lnsg_i$.
    \end{enumerate}
\end{definition}

If $\lnsg_i \in L$ and $\lnsg_{i+1}$ is a child of $c(\lnsg_i)$, then we say that $\lnspath^\circ$ \emph{passes through the leaf $\lnsg_i$}. Note that if $\lnspath^\circ$ is an infinite looping path, then there exists a leaf $\lnsg_i \in L$ through which $\lnspath^\circ$ passes infinitely often. We denote by $\lnspath^\circ_\infty$ the set of leaves in $L$ through which $\lnspath^\circ$ passes infinitely often. Note that $\lnspath^\circ_\infty$ is finite since $\pi$ is finite. Given a cyclic proof $C = (\pi, L, c)$, let
\begin{equation*}
    \operatorname{src}_{C}(v) := \begin{cases}
        c(v) & \text{ if } v \in L\\
        v & \text{ else.}
    \end{cases}
\end{equation*}
We call $\operatorname{src}(v)$ the \emph{source of $v$}. In order to distinguish looping paths from ordinary paths through derivations, we also call the latter \emph{upward paths}. We introduce a pre-order on the set of buds of $\pi$ as follows.

\begin{definition}
    Let $(\pi, L, c)$ be a cyclic proof. The \emph{dependency order} $\preceq$ on $L$ is given by: 
$\lnsg_i \preceq \lnsg_j$ if and only if there exists an upward path from $c(\lnsg_j)$ to $c(\lnsg_i)$. 
A \emph{$\preceq$-maximal leaf} of $L' \subseteq L$ is a leaf $\lnsg_i \in L'$ such that for every leaf $\lnsg_j \in L'$, if $\lnsg_i \preceq \lnsg_j$, then $\lnsg_j \preceq \lnsg_i$.
\end{definition}

The following lemma was proven in~\cite[Lemma 6.10]{Zenger2026}.

\begin{lemma}\label{lem: maximal leaf}
   Let $(\pi, L, c)$ be a cyclic proof and $\lnspath^\circ$ an infinite looping path. Then, every $\preceq$-maximal leaf of $\lnspath^\circ_\infty$ has the same companion.
\end{lemma}

\begin{lemma}\label{lem: stable annotation0}
  Let $(\pi, L, c)$ be a cyclic proof and $\lnspath^\circ$ an infinite looping path. Let $\lnsg_0$ be $\preceq$-maximal in $\lnspath^\circ_\infty$. Then, for every $\lnsg_i \in \lnspath^\circ_\infty$ the upward path from $c(\lnsg_0)$ to $\lnsg_i$ has a continuous focus trace.
\end{lemma}
\begin{proof}
    Recall that by Lemma~\ref{lem: maximal leaf} every $\preceq$-maximal bud in $\lnspath^\circ_\infty$ has the same companion and therefore an upward path from $c(\lnsg_0)$ to $\lnsg_i$ exists for every $\lnsg_i \in \lnspath^\circ_\infty$. Let $\lnsg_i \in \lnspath^\circ_\infty$ be arbitrary and let $\lnsh$ be any LNS on the upward path from $c(\lnsg_0)$ to $\lnsg_i$. We show that $\lnsh$ is focused.

    Suppose that $\lnsh$ lies on the upward path from $c(\lnsg_0)$ to $\lnsg_0$. Since this path has a continuous focus trace, $\lnsh$ is focused. Otherwise $\lnsh$ does not lie on the upward path from $c(\lnsg_0)$ to $\lnsg_0$. Consider the subtree $T$ of $\pi$ rooted at $\lnsh$. By assumption, $\lnsg_i \in T$ and $\lnsg_0 \not \in T$. Since $\lnspath^\circ_\infty$ visits both $\lnsg_0$ and $\lnsg_i$ infinitely often, it must leave $T$ after some visit of $\lnsg_i$. An ordinary upward path cannot leave $T$, therefore the exit must occur at some bud $\lnsf$ of $T$ by moving to its companion $c(\lnsf)$ which occurs outside of $T$. In fact, $c(\lnsf)$ must lie below $\lnsh$, since $T$ is rooted at $\lnsh$, and so $\lnsh$ lies on the upward path from $c(\lnsf)$ to $\lnsf$. This implies that $\lnsh$ is focused. Since $\lnsg_i$ and $\lnsh$ were arbitrary, this finishes the proof.
\end{proof}

\begin{lemma}\label{lem: stable annotation}
Let $(\pi, L, c)$ be a cyclic proof and $\lnspath^\circ$ an infinite looping path. Then, $\lnspath^\circ$ has a suffix with a continuous focus trace.
\end{lemma}

\begin{proof} Let $\lnspath^\circ = \set{\lnsg_i}_{i \in \mathbb{N}}$ and let $\lnsg$ be $\preceq$-maximal in $\lnspath^\circ_\infty$. Let $j$ be the least natural number such that $\lnsg_j = \lnsg$ and the suffix $\lnspath$ of $\lnspath^\circ$ starting at $\lnsg_j$ only passes through leaves in $\lnspath^\circ_\infty$. Note that $\lnspath$ remains in the subtree $T$ of $\pi$ rooted at $c(\lnsg)$: ordinary upward steps remain in $T$ and, by Lemma~\ref{lem: maximal leaf}, each companion of a bud through which $\lnspath$ passes lies in $T$ as well. Moreover, every LNS occurring in $\lnspath$ must lie on the upward path from $c(\lnsg)$ to $\lnsf$ for some $\lnsf \in \lnspath^\circ_\infty$: otherwise either $\lnspath$ ends in an axiom contradicting our assumption that $\lnspath^\circ$ is infinite, or $\lnspath$ passes through a bud outside of $\lnspath^\circ_\infty$, contradicting our construction of $\lnspath$. By Lemma~\ref{lem: stable annotation0}, every LNS in $\lnspath$ has a formula in focus and so $\lnspath$ has a continuous focus trace.
\end{proof}

\begin{definition}
    Let $C=(\pi, L, c)$ be a cyclic proof and $\lnspath^\circ = (v_i)_{i \in \mathbb{N}}$ an infinite looping path. A \emph{progress point} of $\lnspath^\circ$ is an index $i$ such that $v_{i+1}$ is the right premise of an instance of $\untilli$ or $\untillii$, and the principal formula in $\operatorname{src}_{C}(v_i)$ is in focus.
\end{definition}

\begin{lemma}\label{lem: infinite progress points}
    Let $(\pi, L, c)$ be a cyclic proof and $\lnspath^\circ$ an infinite looping path. There are infinitely many progress points in $\lnspath^\circ$.
\end{lemma}
\begin{proof}
    Mark every edge in the finite tree $\pi$ that passes from the conclusion of a rule instance of $\untilli$ or $\untillii$ with focused principal formula to the right premise. For each node $v$ of $\pi$, define
    \begin{equation*}
        \mu(v) := \text{ the number of marked edges on the upward path from the root to }v.
    \end{equation*}
    Note that $\mu(v) \in \mathbb{N}$ for every $v$ of $\pi$. For every bud $b \in L$, the companion $c(b)$ is a proper ancestor of $b$ and the upward path from $c(b)$ to $b$ passes through at least one progress point. Therefore $(1)$ $\mu(c(b)) \leq \mu(b) - 1$. Suppose, toward a contradiction, that there exists an infinite looping path $\lnspath^\circ = \{v_i\}_{i \in \mathbb{N}}$ which contains only finitely many progress points. Let $N \in \mathbb{N}$ such that no index $i \geq N$ is a progress point. Then for every $i \geq N$ the edge $\operatorname{src}_{C}(v_i) \to v_{i+1}$ is unmarked. Therefore $(2)$ $\mu(\operatorname{src}_{C}(v_i)) = \mu(v_{i+1})$. If $v_i \not \in L$, then $\operatorname{src}_{C}(v_i) = v_i$ and therefore $(2)$ gives $\mu(v_i) = \mu(v_{i+1})$. If $v_i \in L$, then $\operatorname{src}_{C}(v_i) = c(v_i)$ and so $(1)$ and $(2)$ together give
    \begin{equation*}
        \mu(v_{i+1}) = \mu(c(v_i))  \leq \mu(v_i) -1.
    \end{equation*}
    Note that $\lnspath^\circ$ passes infinitely often through a bud. Hence $\{\mu(v_i)\}_{i \geq N}$ is a decreasing sequence of natural numbers that strictly decreases infinitely often. This contradicts the well-foundedness of the natural numbers. Hence $\lnspath^\circ$ contains infinitely many progress points.
\end{proof}

Now we turn toward proving the Shifting Lemma (see Lemma~\ref{lem:shifting} below). We begin by introducing a notion of isomorphism between cyclic proofs. Let $\der$ and $\der'$ be derivations in $\calc$. A \emph{homomorphism} from $\der$ to $\der'$ is a function $h$ mapping occurrences of LNSs in $\der$ to occurrences of LNSs in $\der'$ such that whenever a rule $\ru$ occurs in $\der$ with premises $\lnsg_{1},\ldots,\lnsg_{n}$ and conclusion $\lnsg$, the same rule $\ru$ occurs in $\der'$ with premises $h(\lnsg_{1}),\ldots,h(\lnsg_{n})$ and conclusion $h(\lnsg)$.\footnote{Note that homomorphisms map \emph{occurrences} of linear nested sequents, so two distinct occurrences of the same LNS $\lnsg$ may map to different LNSs $\lnsh$ and $\lnsf$.} We say that $\der$ and $\der'$ are \emph{isomorphic} if there exists a homomorphism $h:\der\to\der'$ whose inverse $h^{-1}$ exists and is also a homomorphism, and call $h$ an \emph{isomorphism}. Isomorphisms preserve the structure of derivations, but not necessarily their \emph{content}: isomorphic derivations may be derivations of different LNSs. A \emph{strict isomorphism} is an isomorphism $h: \pi \to \pi'$ which satisfies the following two properties for any $\lnsg$ in $\pi$: (1) $\lnsg$ is saturated iff $h(\lnsg)$ is saturated and (2) $\last{\lnsg} = \last{h(\lnsg)}$; in such a case, we call $\pi$ and $\pi'$ \emph{strictly isomorphic}. We call cyclic proofs $(\pi, L, c)$ and $(\pi', L', c')$ \emph{strictly isomorphic} if there exists a strict isomorphism $\shift: \prf \to \prf'$ such that (1) $\lnsf \in L$ iff $\shift(\lnsf) \in L'$ and (2) for any $\lnsf \in L$, $h(c(\lnsf)) = c'(h(\lnsf))$.

\begin{example} To illustrate the unraveling procedure, we provide an example. Consider the cyclic proof $(\prf,L,c)$ in Figure~\ref{fig:example-cyc-proof} whose conclusion is $c(\lnsg) = \lnsh =  p \until q^\annof \sar p \until q$. As shown in Figure~\ref{fig:example-unravel}, starting from the bud $\lnsg$, we shift the rule applications forward by one component to obtain a proof $(\prf',L',c')$ of $\lnsg$. This procedure can be repeated on the bud $p \sar q \sep p \sar q \sep p \until q^\annof \sar p \until q$, and hence iterated ad infinitum to extract a non-wellfounded proof. Observe that $(\prf,L,c)$ and $(\prf',L',c')$ are strictly isomorphic.
\end{example}

We now introduce notation used in the lemma below. A \emph{context} is an LNS with a \emph{hole} $\nshole$, which replaces a single Gentzen sequent in an LNS (cf.~\cite{Bru09,Lel15}). For example, $\lnsg\nshole = A,B \sar \sep \nshole \sep C \sar D$ is a context. In a context $\lnsg\nshole$, substituting an LNS $\lnsh$ for the hole yields the LNS $\lnsg\set{\lnsh}$.\footnote{Technically, $\lnsg\set{\lnsh}$ is only a well-defined LNS if the conditions placed on the annotations are satisfied. For example if both $\lnsg\nshole$ and $\lnsh$ were to contain a formula in focus, then $\lnsg\set{\lnsh}$ would contain two formulae in focus and thus would not be a well-defined LNS. We ignore this issue here, as this situation does not occur in any of the following proofs.} For example, substituting $\lnsh = E \sar F \sep G \sar$ in the context $\lnsg\nshole$ above yields $\lnsg\set{\lnsh} = A,B \sar \sep E \sar F \sep G \sar \sep C \sar D$. We may write $\lnsg\set{\Gamma \sar \Delta}_{i}$ to indicate that $\Gamma \sar \Delta$ is the $i$-component of $\lnsg$ and $\lnsg\set{\Gamma \sar \Delta \sep \Sigma \sar \Pi}_{i}$ to indicate that $\Gamma \sar \Delta$ is the $i$-component and $\Sigma \sar \Pi$ is the $i+1$-component of $\lnsg$.

\begin{figure}
\begin{center}
\begin{tabular}{c c}
$\prf_{1}' \ :=$
&
\AxiomC{}
\RightLabel{$\id$}
\UnaryInfC{$p \sar q \sep q \sar p, q$}
\AxiomC{}
\RightLabel{$\id$}
\UnaryInfC{$p \sar q \sep q \sar q \sep \sar p \until q$}
\RightLabel{$\untilrii$}
\BinaryInfC{$p \sar q \sep q \sar p \until q$}
\DisplayProof
\end{tabular}
\end{center}

\begin{center}
\begin{tikzpicture}
\node[] (a) [] {\AxiomC{$\prf_{1}'$}

\AxiomC{}
\RightLabel{$\id$}
\UnaryInfC{$p \sar q \sep p \sar p, q \sep p \until q^\annof \sar$}

\AxiomC{$p \sar q \sep p \sar q \sep p \until q^\annof \sar p \until q$}
\RightLabel{$\untilri$}
\BinaryInfC{$p \sar q \sep p \sar p \until q \sep p \until q^\annof \sar$}
\RightLabel{$\untillii$}
\BinaryInfC{$p \sar q \sep p \until q^\annof \sar p \until q$}
\DisplayProof};


\draw[->,rounded corners=.25cm,dashed]
  (2.2,.75) -- (2.2,1) -- (7.2,1) -- node[right,midway]{$c$} 
  (7.2,-.65) -- (.2,-.65);

\end{tikzpicture}
\end{center}

\caption{Shifting Lemma Example.\label{fig:example-unravel}}
\end{figure}

\begin{lemma}[Shifting Lemma]\label{lem:shifting}
Let $\tuple{\der,L,c}$ be a cyclic proof of $\lnsg$ with $\lnsg = c(\lnsg')$ and $\lnsg' \in L$. Then, there exists a strictly isomorphic cyclic proof $\tuple{\der',L',c'}$ of $\lnsg'$.
\end{lemma}

\begin{proof} Let $\lnsg = \lnsh \sslash \Gamma \sar \Delta$ and $\lnsg' = \lnsh \sslash \lnsk \sslash \Gamma \sar \Delta$ such that $\size{\lnsk} \geq 1$, $\size{\lnsg} = n+1$, and $\size{\lnsg'} = k+1$. By definition, $\lnsg$ and $\lnsg'$ are saturated. We simultaneously define a derivation $\pi'$ of $\lnsg'$ and a strict isomorphism $\shift$ from LNSs in $\prf$ to LNSs in $\pi'$ in a stepwise manner. We refer to $\prf$ as the `input derivation' and to $\prf'$ as the `output derivation'. We show by induction on the depth $d$ of $\lnsf$ in $\prf$ that the following invariant (I) holds, which is a conjunction of three statements: (a) $h$ is an isomorphism up to depth $d$, (b) $\lnsf$ is saturated \iffi $\shift(\lnsf)$ is saturated, and (c) if $i \geq n$, then $\lnsf(i) = \shift(\lnsf)(j)$ with $j = k + (i - n)$.

We remark that condition (c) ensures that components at or
beyond the index $n$ in LNSs of the input derivation are preserved
via $h$ in the output derivation, but shifted forward by $(k - n)$
positions. This shift reflects the fact that $\lnsg'$ extends $\lnsg$ with the additional components of $\lnsk$. Moreover, since $\lnsg$ is saturated, rule applications in the input derivation are confined to $i$-components with $i \geq n$, and these are precisely the components that are relocated from the index $i$ to the index $k + (i - n)$ in the output.

\case{Base case.} We let the derivation of $\lnsg'$ be the tree consisting of the single node $\lnsg'$ and set $\shift(\lnsg) := \lnsg'$. Note that (I) holds by the shape of $\lnsg$ and $\lnsg'$.


\case{Inductive step.} We simultaneously treat every node $\lnsf$ of depth $d\geq0$ in $\prf$ and make a case distinction on the rule applied to $\lnsf$ in $\prf$. We show the $\impr$ and $\untilli$ cases as the remaining cases are similar. As noted above, since $\lnsg$ is saturated, all rule applications occur at $i$-components with $i \geq n$.

\case{Case for $\impr$.}Suppose $\impr$ is applied bottom-up (shown below left) to the $i$-component of $\lnsf$ in $\prf$ with $d$ the depth of $\lnsf$ in $\prf$. By IH, we have constructed the output derivation and map $\shift$ up to depth $d$ and know that the $i$-component of $\lnsf$ is equal to the $j$-component of $\shift(\lnsf)$. We extend the output derivation by bottom-up applying $\impr$ to the $j$-component of $\shift(\lnsf)$, as shown below right.
\begin{center}
\begin{tabular}{c c}
\AxiomC{$\lnsf\set{\Sigma, A \sar B, \Pi}_{i}$}
\RightLabel{$\impr$}
\UnaryInfC{$\lnsf\set{\Sigma \sar A \imp B, \Pi}_{i}$}
\DisplayProof

&

\AxiomC{$\shift(\lnsf)\set{\Sigma, A \sar B, \Pi}_{j}$}
\RightLabel{$\impr$}
\UnaryInfC{$\shift(\lnsf)\set{\Sigma \sar A \imp B, \Pi}_{j}$}
\DisplayProof
\end{tabular}
\end{center}
We set $\shift(\lnsf\set{\Sigma, A \sar B, \Pi}_{i}) := \shift(\lnsf)\set{\Sigma, A \sar B, \Pi}_{j}$. Note that $\shift$ is an isomorphism between the input and output derivations up to depth $d+1$, and so (a) is satisfied. Moreover, for all $m \neq i$ with $m \geq n$, the $m$-component has been unaffected by the $\impr$ application, so the $m$-component of $\lnsf\set{\Sigma, A \sar B, \Pi}_{i}$ and the $k+(m-n)$-component of $\shift(\lnsf)\set{\Sigma, A \sar B, \Pi}_{j}$ are equal. Also, observe that the $i$-component of $\lnsf\set{\Sigma, A \sar B, \Pi}_{i}$ is equal to the $j$-component of $\shift(\lnsf)\set{\Sigma, A \sar B, \Pi}_{j}$. This implies that (b) and (c) hold as well and so that (I) holds.

\case{Case for $\untilli$.} Suppose $\untilli$ is applied bottom-up to the $i$-component of $\lnsf$ in $\prf$ with $d$ the depth of $\lnsf$ in $\prf$, where $\anno \in \set{\annou, \annof}$.
\begin{center}
    \AxiomC{$\lnsf\set{\Gamma, B \sar \Delta \sep \Sigma \sar \Pi}_i$}
    \AxiomC{$\lnsf\set{\Gamma, A  \sar \Delta \sep \Sigma, A \until B^\anno \sar \Pi}_i$}
    \BinaryInfC{$\lnsf\set{\Gamma, A \until B^\anno \sar \Delta \sep \Sigma \sar \Pi}_i$}
    \DisplayProof
\end{center}
By IH we have constructed the output derivation and map $\shift$ up to depth $d$ and know that the $i$- and $i+1$-components of $\lnsf$ are equal to the $j$- and $j+1$-components of $\shift(\lnsf)$, respectively. We extend the output derivation by bottom-up applying $\untilli$ to the $j$-component of $\shift(\lnsf)$, as shown below.
\begin{center}
    \AxiomC{$h(\lnsf)\set{\Gamma, B \sar \Delta \sep \Sigma \sar \Pi}_j$}
    \AxiomC{$h(\lnsf)\set{\Gamma, A  \sar \Delta \sep \Sigma, A \until B^\anno \sar \Pi}_j$}
    \BinaryInfC{$h(\lnsf)\set{\Gamma, A \until B^\anno \sar \Delta \sep \Sigma \sar \Pi}_j$}
    \DisplayProof
\end{center}
Note that if $\anno = \annof$, then both output premises are well-formed, since the corresponding input premises are well-formed and the construction preserves their annotations while inserting only unfocused components. We set $\shift(\lnsf\set{\Gamma, B \sar \Delta \sep \Sigma \sar \Pi}_i):= h(\lnsf)\set{\Gamma, B \sar \Delta \sep \Sigma \sar \Pi}_j$ and $\shift(\lnsf\set{\Gamma, A  \sar \Delta \sep \Sigma, A \until B^\anno \sar \Pi}_i) := h(\lnsf)\set{\Gamma, A  \sar \Delta \sep \Sigma, A \until B^\anno \sar \Pi}_j$. Note that $h$ is an isomorphism between the input and output derivations up to depth $d+1$, and so (a) is satisfied. Moreover, by an argument similar to the $\impr$-case, one can verify that conditions (b) and (c) are satisfied as well and so that (I) holds. This concludes the construction of the output derivation $\prf'$ and of $\shift$.

We now show how to define a cyclic proof of $\lnsg'$ based on $\prf'$ that is strictly isomorphic to $(\pi, L, c)$. First, observe that (I) implies that $h$ is a strict isomorphism between $\pi$ and $\pi'$. Let $(c(\lnsf),\lnsf)$ form a cycle in $(\prf, L, c)$. By the definition of $\shift$, we know that a path occurs from $\shift(c(\lnsf))$ to $\shift(\lnsf)$ in $\prf'$ with $h(\lnsf)$ a leaf of $\pi'$ and that since $\size{c(\lnsf)} < \size{\lnsf}$, the path from $c(\lnsf)$ to $\lnsf$ passes through an expansion rule\footnote{If the expansion rule is $\untilli$ or $\untillii$, then the path passes through the right premise. We will not mention this again.} and so the path from $h(c(\lnsf))$ to $h(\lnsf)$ passes through an expansion rule as well, implying that $\size{h(c(\lnsf))} < \size{h(\lnsf)}$. Moreover, by the invariant (I), we know that $h(\lnsf)$ and $h(c(\lnsf))$ are saturated since $\lnsf$ and $c(\lnsf)$ are saturated, $\last{\lnsf} = \last{h(\lnsf)}$, and $\last{c(\lnsf)} = \last{h(c(\lnsf))}$. Therefore, since $\last{c(\lnsf)} = \last{\lnsf}$, we have $\last{h(c(\lnsf))} = \last{h(\lnsf)}$. Since $(c(\lnsf), \lnsf)$ is a cycle in a cyclic proof, the path from $c(\lnsf)$ to $\lnsf$ has a continuous focus trace and passes through the right premise of a rule $\ru \in \set{\untilli,\untillii}$ with the formula in focus principal. By construction of $\prf'$, $h$ and by (I), the path from $h(c(\lnsf))$ to $h(\lnsf)$ therefore has a continuous focus trace and passes through the right premise of a rule $\ru \in \set{\untilli,\untillii}$ with the formula in focus principal. Let $L' := \set{h(\lnsf) \mid \lnsf \in L}$ and let $c'$ map each $h(\lnsf) \in L'$ to $h(c(\lnsf))$. For any other leaf $h(\lnsh)$ of $\pi'$, $\lnsh$ is a leaf of $\pi$ which does not belong to $L$ and therefore is an initial sequent. By (I), $h(\lnsh)$ is an initial sequent as well; therefore, $(\pi', L', c')$ is a cyclic proof of $\lnsg'$ that is strictly isomorphic to $(\pi, L,c)$.
\end{proof}

Next, we use Lemma~\ref{lem:shifting} to obtain an infinite sequence of `growing' cyclic proofs and define the limit of this sequence. Let $(\pi_0, L_0, c_0)$ be a cyclic proof of $\lnsg$ with $L_0 \not = \emptyset$. Let $\lnsh_0 = c_{0}(\lnsh_1)$ be a companion in $\pi_0$ with \emph{least depth}, i.e., a companion node such that no other companion occurs in $\pi$ with depth strictly smaller (where depth means the distance of $\lnsh$ from the root). Note that a companion of least depth has no companion as a proper ancestor, so its rooted subtree is a self-contained cyclic proof. Thus let $(\pi, L, c)$ denote the cyclic proof rooted at $\lnsh_0$ in $(\pi_0, L_0, c_0)$.

By Lemma~\ref{lem:shifting}, we can transform $\tuple{\der,L,c}$ into a strictly isomorphic cyclic proof $\tuple{\der',L',c'}$ of $\lnsh_{1}$ such that $\lnsh_{1} = c'(\lnsh_{2})$ for some leaf $\lnsh_{2}$ in $L'$. We construct a new cyclic proof $\tuple{\der_1,L_1,c_1}$ of $\lnsg$ as follows: first, we `paste' $\der'$ above $\lnsh_{1}$ in $\der_0$ to obtain a new derivation $\der_1$. Second, we set $L_1 := (L_0 \setminus \set{\lnsh_{1}}) \cup L'$. Third, we define $c_1$ as follows:
\begin{equation*}
c_1(\lnsf) = 
\begin{cases} 
c_0(\lnsf) & \text{if $\lnsf \in L_0 \setminus \set{\lnsh_{1}}$,} \\
c'(\lnsf) & \text{if $\lnsf \in L'$.}
\end{cases}
\end{equation*}
We first argue that $\tuple{\der_1, L_1, c_1}$ is a cyclic proof of $\lnsg$. Indeed, $\der_1$ is a finite derivation of $\lnsg$ since $\der_0$ is a finite derivation of $\lnsg$ and $\der'$ is a finite derivation of $\lnsh_1$. For condition~(3) of Definition~\ref{dfn:cyclic-proof}, consider first a bud
$\lnsf \in L_0 \setminus \{\lnsh_1\}$. Since the path from the root to $\lnsf$ in $\der_1$ is identical to the corresponding path in $\der_0$, the cycle $(c_0(\lnsf), \lnsf)$ satisfies all required conditions in $\der_1$. Next, consider a bud $\lnsf \in L'$. Since $\tuple{\der', L', c'}$ is a cyclic proof, the path from $c'(\lnsf)$ to $\lnsf$ in $\der'$ is preserved in $\der_1$ and satisfies all required conditions. Finally, every leaf of $\der_1$ not in
$L_1$ is either a leaf of $\der_0$ not in $L_0$ or a leaf of $\der'$ not in $L'$, and is therefore an initial sequent in either case.

Note that the cycle from $\lnsh_{1}$ to $\lnsh_{0}$ in $(\pi_0, L_0, c_0)$ has been shifted upward and replaced by a cycle from $\lnsh_{2}$ to $\lnsh_{1}$ in $\tuple{\der_1,L_1,c_1}$. In general, $\lnsh_0$ may be the companion of several buds, so after repeating the above process for each bud, $\lnsh_0$ is no longer a companion in the generated cyclic proof. By successively repeating the above process for companions of least depth, we obtain an infinite sequence $\tuple{\der_0,L_0,c_0}$, $\tuple{\der_1,L_1,c_1}$, $\ldots$ of cyclic proofs.  

By construction, each $\der_i$ is a tree of LNSs with root $\lnsg$, i.e., $\der_i=(T_i, \leq_i,\lnsg)$ where $T_i$ is a set of nodes (labeled by LNSs) and $(T_i,\leq_i)$ is a poset, such that for each $t \in T_i$, the down-set of $t$ in $T_i$ is well-ordered. Observe that for all $i \in \mathbb{N}$, $T_i \subseteq T_{i+1}$ and $\leq_i \ \subseteq \ \leq_{i+1}$. We define the \emph{limit derivation} of the sequence of cyclic proofs to be $\pi_\omega := (T_\omega, \leq_\omega, \lnsg)$ where $T_\omega := \bigcup_{i \in \mathbb{N}} T_i$ and $\leq_\omega \ := \bigcup_{i \in \mathbb{N}} \leq_i$. Note that $\pi_\omega$ is a well-defined derivation of $\lnsg$ in $\calci$. 

Note that in each step of the unfolding construction, one bud $\lnsg$ is unfolded, and all other existing buds whose companion is $c(\lnsg)$ remain unchanged. If the selected companion has depth $d$, every newly introduced companion lies in the subtree rooted at $\lnsg$ and therefore has depth strictly greater than $d$. Together with the least-depth scheduling rule, this yields the following result.

\begin{lemma}\label{lem: unravelling axioms}
 Every leaf of $\pi_\omega$ is an initial LNS.   
\end{lemma}
\begin{proof}
    Let $u$ be a leaf of $\pi_\omega$, labelled by $\lnsh$, and choose $i\in\mathbb N$ such that $u$ is a leaf of $\pi_i$. Then either $\lnsh$ is initial or $u\in L_i$. Suppose $u\in L_i$. Since $u$ remains a leaf of $\pi_\omega$, it is never unfolded. Consequently, the definition of the construction implies
\begin{equation*}
u\in L_j
\quad\text{and}\quad
c_j(u)=c_i(u)
\qquad\text{for every }j \geq i.
\end{equation*}
Let $D$ be the depth of $c_i(u)$, and let $d_j$ be the least depth of a companion in $(\pi_j, L_j, c_j)$. Since $c_i(u)$ remains a companion, $d_j\leq D$ for every $j\geq i$.
At each step, existing nodes retain their depths, and every newly introduced companion has depth strictly greater than that of the selected companion. Since the selected companion has least depth, it follows that $d_{j+1}\geq d_j$. Thus $\{d_j\}_{j\geq i}$ is a bounded non-decreasing sequence of natural numbers. Therefore there exists $J\geq i$ and $d\in\mathbb N$ such that $d_j=d$ for every $j\geq J$. There are only finitely many buds in $(\pi_J, L_J, c_J)$ whose companions have depth $d$. From stage $J$ onwards, each unfolding step removes exactly one such bud and introduces no new bud whose companion has depth $d$. Hence, after finitely many steps, no companion at depth $d$ remains, contradicting $d_j=d$ for every $j\geq J$.
Therefore $u\notin L_i$, and so $\lnsh$ is initial.
\end{proof}

In order to show that every infinite branch $\branch$ of $\pi_\omega$ is progressing, we show that every such $\branch$ \emph{corresponds to} a looping path of $(\pi_0, L_0,c_0)$ in the following sense.

\begin{definition}\label{dfn:corresponds} Given $j,n \in \mathbb{N}$ and finite looping paths $\lnspath_{j}^\circ = \set{\lnsg_i}_{i \leq k}$ of $C_j=(\pi_{j}, L_{j}, c_{j})$ and $\lnspath_n^\circ = \set{\lnsh_i}_{i \leq k}$ of $C_n=(\pi_n, L_n, c_n)$ of equal length, the looping path $\lnspath^\circ_n$ \emph{corresponds to} $\lnspath^\circ_j$ \iffi for all $i \leq k$
\begin{enumerate}

\item[$(1)$] $\operatorname{src}_{C_j}(\lnsg_i)$ is the conclusion of an instance of $\ru \in \calci$ \iffi $\operatorname{src}_{C_n}(\lnsh_i)$ is the conclusion of an instance of $\ru$.

\item[$(2)$] for $i < k$, $\operatorname{src}_{C_j}(\lnsg_i)$ is the conclusion of a binary rule $\mathsf{r}$ and $\lnsg_{i+1}$ is the right (resp. left) premise if and only if $\operatorname{src}_{C_n}(\lnsh_i)$ is the conclusion of $\mathsf{r}$ and $\lnsh_{i+1}$ is the right (resp. left) 
premise;

\item[$(3)$] $\lnsg_i$ contains a formula in focus \iffi $\lnsh_i$ contains a formula in focus, and if $\operatorname{src}_{C_j}(\lnsg_{i})$ is the conclusion of a logical rule $\mathsf{r}$ whose principal formula is in focus, then the principal formula of the rule 
applied to $\operatorname{src}_{C_n}(\lnsh_i)$ is in focus (and vice versa);

\item[$(4)$] $\lnsg_i$ is saturated \iffi $\lnsh_i$ is saturated, and if so, $\last{\lnsg_i} = \last{\lnsh_i}$.


\end{enumerate}
For $\pi_\omega$ define $\operatorname{src}_C(v) := v$ for all nodes $v \in \pi_\omega$. An infinite branch $\branch=\set{\lnsh_i}_{i \in \mathbb{N}}$ of $\pi_\omega$ then \emph{corresponds to} an infinite looping path $\lnspath^\circ = \set{\lnsg_i}_{i \in \mathbb{N}}$ of $(\pi_0, L_0,c_0)$ \iffi for all $i \in \mathbb{N}$ conditions (1)--(4) are satisfied.
\end{definition}

\begin{lemma}\label{lem: corresponding paths} For any $j \in \mathbb{N}$ and every finite looping path $\lnspath^\circ_{j+1}$ of $(\pi_{j+1}, L_{j+1}, c_{j+1})$, there exists a looping path $\lnspath^\circ_j$ of $(\pi_j, L_j, c_j)$ of equal length, such that $\lnspath^\circ_{j+1}$ corresponds to $\lnspath^\circ_j$.
\end{lemma}

\begin{proof}
Write $C_i:=(\pi_i,L_i,c_i)$. Suppose that $C_{j+1}$ is obtained from $C_j$ by unfolding the bud $\lnsh \in L_j$. Let  $C=(\pi,L,c)$ be the cyclic subproof rooted at $c_j(\lnsh)$, and let $C'=(\pi',L',c')$ be the cyclic proof pasted above $\lnsh$. Let $h:\pi\longrightarrow\pi'$ be the strict isomorphism constructed in the proof of Lemma~\ref{lem:shifting}. In particular,
\[
h(c(\lnsh))=\lnsh,\qquad
L'=h[L],\qquad
c'(h(\lnsf))=h(c(\lnsf))
\quad\text{for every }\lnsf \in L.
\tag{*}
\]
Recall that besides being a strict isomorphism, the constructed $h$ in the proof of Lemma~\ref{lem:shifting} preserves whether a node contains a formula in focus and whether the principal formula of a logical rule instance is in focus. Indeed, the fixed prefixes in the input and output derivations are unfocused, while invariant (I)(c) and the construction preserve the annotations in all remaining components and copy each inference with the same annotated principal formula.

Let $\lnspath^\circ_{j+1}=(\lnsg_i)_{i\leq k}$ be a finite looping path of $C_{j+1}$. Recall that looping paths start at the root. Therefore the only way for $\lnspath^\circ_{j+1}$ to reach a node in $C'$ is to pass through $\lnsh$. Moreover, after passing through $\lnsh$, the path remains in $C'$: ordinary steps remain in that subtree, and all companions of its buds belong to that subtree, since $c(\lnsh)$ has least depth.

If $\lnsh$ does not occur in $\lnspath^\circ_{j+1}$, set $\lnspath^\circ_j:=\lnspath^\circ_{j+1}$. Since $\lnspath^\circ_{j+1}$ does not enter $C'$ and $C_j$ and $C_{j+1}$ are otherwise identical, $\lnspath^\circ_{j}$ is a looping path of $C_j$ which corresponds to $\lnspath^\circ_{j+1}$.

If $\lnsh$ occurs in $\lnspath^\circ_{j+1}$, let $t \leq k$ be the least natural number such that $\lnsg_t=\lnsh$. Define
\[
\lnspath^\circ_j=(\lnsf_i)_{i\leq k},
\qquad
\lnsf_i:=
\begin{cases}
\lnsg_i & i\leq t,\\[1mm]
h^{-1}(\lnsg_i)& t<i\le k.
\end{cases}
\]
We first verify that $\lnspath^\circ_j$ is a looping path of $C_j$. The initial segment $(\lnsf_i)_{i \leq t}$ only passes through nodes in $\pi_{j+1}$ outside of $\pi'$ which also belong to $\pi_j$. Hence $(\lnsf_i)_{i \leq t}$ is an initial segment of a looping path of $C_j$. If $t<k$, then $\lnsg_{t+1}$ is a child of the root $\lnsh$ of $\pi'$. Since $h(c_j(\lnsh))=\lnsh$, $\lnsf_{t+1}=h^{-1}(\lnsg_{t+1})$ is a child of $c_j(\lnsh)$. Thus the step from $\lnsf_t=\lnsh$ to $\lnsf_{t+1}$ is permitted by Definition~\ref{dfn:looping-path}(3). For every $t<i\leq k$, we have
\begin{equation*}
\operatorname{src}_{C_j}(\lnsf_i)
=
h^{-1}\!\left(
\operatorname{src}_{C_{j+1}}(\lnsg_i)
\right).
\tag{**}
\end{equation*}
For a non-bud this follows directly from the definition of $\operatorname{src}$; for a bud $\lnsf_i$ we have $c'(h(\lnsf_i)) = h(c(\lnsf_i))$ by $(\ast)$. Hence  by applying $h^{-1}$ to both sides and taking $\lnsf_i = h^{-1}(\lnsg_i)$, we obtain  
\begin{equation*}
    h^{-1}(c'(\lnsg_i))
=
c(h^{-1}(\lnsg_i)) .
\end{equation*}
Finally, since $c'(\lnsg_i) = \operatorname{src}_{C_{j+1}}(\lnsg_i)$, $h^{-1}(\lnsg_i) = \lnsf_i \in L_i$ and $c(h^{-1}(\lnsg_i)) = \operatorname{src}_{C_j}(\lnsf_i)$, $(\ast \ast)$ follows.

Consequently, if $t<i<k$, then $\lnsf_{i+1}$ is a child of $\operatorname{src}_{C_j}(\lnsf_i)$, because $\lnsg_{i+1}$ is a child of
$\operatorname{src}_{C_{j+1}}(\lnsg_i)$ and $h^{-1}$ is a homomorphism. This handles both ordinary steps and bud jumps. Moreover, $h$ preserves buds and initial leaves.
Thus $\lnspath^\circ_j$ is a looping path of $C_j$ of the required length.
It remains to check correspondence. We distinguish three cases.

\case{Case $i<t$.} By definition $\lnsf_i = \lnsg_i$ and $\operatorname{src}_{C_j}(\lnsf_i) = \operatorname{src}_{C_{j+1}}(\lnsg_i)$. When $i<k$, the selected premise is unchanged as well. Therefore conditions (1)–(4) hold.

\case{Case $i=t$.} Then $\lnsf_t = \lnsg_t = \lnsh$, but their sources are
\[
\operatorname{src}_{C_{j+1}}(\lnsh)=\lnsh,
\qquad
\operatorname{src}_{C_j}(\lnsh)=c_j(\lnsh).
\]
By construction of $\pi'$ in the Shifting Lemma, the same rule is applied to $\lnsh$ in $\pi'$ and to $c_j(\lnsh)$ in $\pi$. This implies that condition (1) is satisfied. If $t<k$, $\lnsh$ is the conclusion of a binary rule in $\pi'$ and $\lnsg_{t+1}$ is the left (right, respectively) premise, then $h^{-1}(\lnsg_{t})$ is the conclusion of the same binary rule by (1) and $h^{-1}(\lnsg_{t+1})$ is the left (right, respectively) premise by construction of $h$. This implies that condition (2) holds. If $t=k$, condition (2) is vacuous. For (3), $\lnsf_t$ contains a formula in focus \iffi $\lnsg_t$ contains a formula in focus since $\lnsf_t = \lnsg_t = \lnsh$. Moreover, if $\lnsh$ is the conclusion of a logical rule $\ru$ with the principal formula in focus in $\pi'$, then $c_j(\lnsh)$ is the conclusion of $\ru$ with the principal formula in focus in $\pi$, due to the construction of $h$, and vice versa. Therefore, (3) holds. Finally, (4) holds trivially since $\lnsf_t = \lnsg_t = \lnsh$. 

\case{Case $t<i\le k$.} Then $\lnsf_i = h^{-1}(\lnsg_i)$ and $\operatorname{src}_{C_j}(\lnsf_i) =
h^{-1}(
\operatorname{src}_{C_{j+1}}(\lnsg_i))$ by $(\ast \ast)$. Thus $\operatorname{src}_{C_j}(\lnsf_i)$ is the conclusion of an instance of $\ru$ \iffi $\operatorname{src}_{C_{j+1}}(\lnsg_i)$ is the conclusion of an instance of $\ru$ since $h^{-1}$ is a homomorphism. Thus, condition (1) holds. Conditions (2) and (3) follow by similar arguments. Finally, since $h$ is a strict isomorphism,
\[
\lnsg_i\text{ is saturated}
\iff
\lnsf_i\text{ is saturated},
\qquad
\last{\lnsg_i}
=
\last{\lnsf_i},
\]
so condition (4) holds. Therefore, $\mathcal P^\circ_{j+1}$ corresponds to $\mathcal P^\circ_j$.
\end{proof}

\begin{lemma}\label{cor: corresponding paths} For any $j \in \mathbb{N}$ and every finite looping path $\lnspath^\circ_j$ of $(\pi_{j}, L_{j}, c_{j})$, there exists a finite looping path $\lnspath^\circ_0$ of $(\pi_0, L_0, c_0)$ of equal length, such that $\lnspath^\circ_j$ corresponds to $\lnspath^\circ_0$.
\end{lemma}
\begin{proof} By induction on $j$. The base case for $j = 0$ is trivial. Thus, suppose $j = n+1$ and consider a finite looping path $\lnspath^\circ_{n+1}$ of $(\pi_{n+1}, L_{n+1}, c_{n+1})$. By Lemma~\ref{lem: corresponding paths}, there exists a finite looping path $\lnspath^\circ_n$ of $(\pi_n, L_n, c_n)$ of equal length, such that $\lnspath^\circ_{n+1}$ and $\lnspath^\circ_n$ are corresponding. By the induction hypothesis, there exists a finite looping path $\lnspath^\circ_0$ of $(\pi_0, L_0, c_0)$ of equal length, such that $\lnspath^\circ_n$ and $\lnspath^\circ_0$ are corresponding. Since the notion of correspondence between looping paths is transitive (each condition is a biconditional that composes), $\lnspath^\circ_{n+1}$ and $\lnspath^\circ_0$ are corresponding as well.
\end{proof}

We are now ready to prove the main theorem of this section.

\begin{theorem}\label{thm:cyclic-to-non-wf}
Every $\calccyc$-proof can be transformed into an $\calci$-proof satisfying the SRP.
\end{theorem}

\begin{proof} Let $\tuple{\der_0,L_0,c_0}$ be a cyclic proof of $\lnsg$. If $L_0 = \emptyset$, then $\der_0$ is a finite $\calci$-proof of $\lnsg$ which vacuously satisfies the SRP. Thus, suppose that $L_0 \not = \emptyset$ and let $\pi_\omega$ be the infinite $\calci$-derivation obtained by successively applying the Shifting Lemma to least-depth companions as described above. We will now argue that $\pi_\omega$ is a $\calci$-proof. By Lemma~\ref{lem: unravelling axioms} every leaf of $\pi_\omega$ is an initial sequent. Thus it remains to show that every infinite branch is progressing.

Suppose $\branch = \{\lnsg_m\}_{m \in \mathbb{N}}$ is an infinite branch of $\pi_\omega$. Note that by construction of $\pi_\omega$, for every $j \in \mathbb{N}$ there exists a natural number $m_j$ such that $\lnsg_{m_j}$ is a bud of $C_j=(\pi_j, L_j, c_j)$ and the initial segment $\branch_j = \set{\lnsg_m}_{m \leq m_j}$ of $\branch$ is a branch of $\pi_j$. We now define an infinite looping path $\lnspath^\circ = \set{\lnsh_i}_{i \in \mathbb{N}}$ of $C_0=(\pi_0, L_0, c_0)$ such that $\lnspath^\circ$ corresponds to $\branch$. To do so, we construct an infinite sequence of finite looping paths $\lnsq^\circ_i$ of $C_0$. We proceed by induction on $i$ and show that invariant (II) holds, which is a conjunction of two statements:
\begin{itemize}
    \item[(i)] $\lnsq_{i}^\circ$ contains $\lnsq_{i-1}^\circ$ as an initial segment (for $i > 0$) and
    \item[(ii)] $\lnsq^\circ_i$ corresponds to $\branch_i$.
\end{itemize} 
Let $\lnsq_0^\circ := \branch_0$. By our previous observation $\lnsq_0^\circ$ is a branch of $\pi_0$ ending in a bud. Since $i=0$, (i) holds vacuously, and (ii) holds since $\lnsq_0^\circ = \branch_0$. Therefore, invariant (II) holds. Now suppose we have constructed $\lnsq_{i}^\circ$ and (II) holds. Since $\branch_{i+1}$ is a branch of $\pi_{i+1}$ ending in a bud, it is also a looping path of $C_{j+1}$ and therefore corresponds to a looping path $\lnsq^\circ$ of $C_0$ by Lemma~\ref{cor: corresponding paths}. Set $\lnsq^\circ_{j+1} := \lnsq^\circ$. Note that (ii) holds. For (i) we show that $\lnsq^\circ_j = \{\lnsh_m\}_{m \leq m_j}$ is an initial segment of $\lnsq^\circ_{j+1}=\{\lnsf_m\}_{m \leq m_{j+1}}$. First note that both $\lnsq^\circ_{j}$ and $\lnsq^\circ_{j+1}$ start at the root of $\pi_0$. Therefore $\lnsh_0 = \lnsf_0$. Now suppose inductively that $\lnsh_m = \lnsf_m$ for $m < m_j$. Therefore $\operatorname{src_{C_0}}(\lnsh_m) = \operatorname{src_{C_0}}(\lnsf_m)$. Recall that $\lnsq^\circ_j$ and $\lnsq^\circ_{j+1}$ correspond to $\branch_j = \{\lnsg_m\}_{m \leq m_j}$ and $\branch_{j+1} = \{\lnsg'_m\}_{m \leq m_{j+1}}$, respectively. Thus at $\lnsh_m$ the looping path  $\lnsq^\circ_j$ chooses the same premise as the transition in $\branch_j$ from $\lnsg_m$ to $\lnsg_{m+1}$, and similarly for $\lnsq^\circ_{j+1}$ and $\branch_{j+1}$. Since $m < m_j$ and $\branch_j$ is an initial segment of $\branch_{j+1}$, $\lnsg_{m+1} = \lnsg'_{m+1}$ and therefore $\lnsh_{m+1} = \lnsf_{m+1}$. Thus $\lnsh_m = \lnsf_m$ for all $m \leq m_j$, which implies that $\lnsq^\circ_j$ is an initial segment of $\lnsq^\circ_{j+1}$ and so (i) holds. We therefore conclude that for any $i \in \mathbb{N}$, invariant (II) holds. Finally, we let 
\begin{equation*}
    \lnspath^\circ := \bigcup_{i \in \mathbb{N}} \lnsq^\circ_i.\footnote{We view a sequence $\set{\lnsg_i}_{i \in I}$ of LNSs as a function which maps natural numbers in $I$ to LNSs.}
\end{equation*} 
Note that $\lnspath^\circ$ is a well-defined infinite looping path of $(\pi_0, L_0, c_0)$ by (i) and corresponds to $\branch$ by (ii). We now argue that $\lnspath^\circ$ is progressing.

Write $\lnspath^\circ = \{\lnsh_m\}_{m \in \mathbb{N}}$. By Lemma~\ref{lem: stable annotation} there exists $N \in \mathbb{N}$ such that the suffix $\{\lnsh_m\}_{m \geq N}$ of $\lnspath^\circ$ has a continuous focus trace. By Lemma~\ref{lem: infinite progress points}, $\lnspath^\circ$ has infinitely many progress points. Thus the suffix $\{\lnsh_m\}_{m \geq N}$ of $\lnspath^\circ$ has infinitely many progress points as well. Therefore $\{\lnsh_m\}_{m \geq N}$ is a progressing suffix, which implies that $\lnspath^\circ$ is progressing.

Since $\lnspath^\circ$ corresponds to $\branch$, we can transfer the progressing property from $\lnspath^\circ$ to $\branch$ using Definition~\ref{dfn:corresponds}. Let $\lnspath^\circ_{\geq j} = \set{\lnsh_i}_{i \geq j}$ be the progressing suffix of $\lnspath^\circ$ and let $\branch_{\geq j} = \set{\lnsg_i}_{i \geq j}$ be the corresponding suffix of $\branch$. By condition~(3), $\lnsg_i$ contains a formula in focus if and only if $\lnsh_i$ contains a formula in focus for each $i \geq j$. Since $\lnspath^\circ_{\geq j}$ has a continuous focus trace, $\branch_{\geq j}$ has a continuous focus trace as well. By conditions~(1) and~(2), whenever $\operatorname{src}_{C_0}(\lnsh_i)$ is the conclusion of a rule $\mathsf{r} \in \set{\untilli, \untillii}$ and $\lnsh_{i+1}$ is the right premise, we have that $\lnsg_i$ is the conclusion of the same rule $\mathsf{r}$ and $\lnsg_{i+1}$ is the right premise. Moreover, by condition~(3), if the principal formula of $\mathsf{r}$ applied to $\operatorname{src}_{C_0}(\lnsh_i)$ is in focus, then the principal formula of $\mathsf{r}$ applied to $\lnsg_i$ is in focus. Since $\lnspath^\circ$ passes infinitely often through the right premise of $\mathsf{r} \in \set{\untilli, \untillii}$ with the formula in focus principal, the same holds for $\branch$. We thus conclude that $\branch$ is progressing.

Finally, we argue that $\pi_\omega$ satisfies the SRP. By condition~(4), whenever $\lnsh_i$ is saturated, $\lnsg_i$ is saturated as well. Since $\lnspath^\circ$ passes infinitely often through the buds $\lnsh \in \lnspath^\circ_\infty$ which are saturated by Definition~\ref{dfn:cyclic-proof}--(3), $\lnspath^\circ$ passes through infinitely many saturated LNSs. Therefore, $\branch$ passes through infinitely many saturated LNSs as well. Since this holds for every infinite branch $\branch$ of $\pi_\omega$, $\pi_\omega$ satisfies the SRP. We conclude that $\pi_\omega$ is an $\calci$-proof of $\lnsg$ which satisfies the SRP.
\end{proof}

We obtain the soundness of the cyclic LNS system as a corollary of Theorems~\ref{thm:soundness} and~\ref{thm:cyclic-to-non-wf}. 

\begin{corollary}[Soundness of $\calccyc$]\label{cor:soundness-cyclic-proofs}
If $\lnsg$ has an $\calccyc$-proof, then $\lnsg$ is valid.
\end{corollary}




\section{Concluding Remarks}\label{sec:conclusion}

There are several directions for future research. First, it would be natural to investigate whether---and under what conditions---the proof transformations between non-wellfounded and cyclic proofs developed here can be extended to more expressive multisequent formalisms, such as nested and labeled sequents. Such generalizations would deepen our understanding of cycle recognition and unraveling beyond the linear nested setting. Second, while non-wellfounded Gentzen-style sequent calculi for $\ltl$ were given in~\cite{AloEtAl20,BruLan08,KokStu16}, the problem of syntactic cut-elimination was left open in all three works. We conjecture that the additional structural expressiveness of $\calci$ enables a syntactic cut-elimination proof, which we aim to investigate. Third, it would be interesting to study proof transformations between our linear nested sequent calculi and the existing Gentzen and tableau systems of~\cite{AloEtAl20,BruLan08,KokStu16,Rey16}. Such translations would clarify the relationship between these formalisms and allow results to be transferred between them---in particular, should cut-elimination be established for $\calci$, such 
translations could provide a pathway to resolving the analogous open problem for the Gentzen systems of~\cite{AloEtAl20,BruLan08,KokStu16}. Finally, while our contribution is primarily theoretical, the invertible and formula-driven nature of our cyclic system suggests applications to automated proof-search and proof certificates for $\ltl$.


\bibliographystyle{eptcs}
\bibliography{bibliography}









\end{document}